\documentclass{article}
\usepackage[a4paper, margin=2.7cm]{geometry}
\usepackage[OT1, T1]{fontenc}
\DeclareTextSymbolDefault{\DH}{T1}
\usepackage[utf8]{inputenc}
\usepackage[english]{babel}

\usepackage{refcount}
\usepackage{xargs}
\usepackage{xifthen}
\usepackage{enumitem}
\setlist{nolistsep, noitemsep}
\usepackage{amsthm}
\usepackage{amsmath}
\usepackage{amsfonts}
\DeclareMathOperator*{\Plus}{+}

\usepackage{calc}
\usepackage{array}
\usepackage[ruled,noresetcount]{algorithm2e}
\usepackage{color, colortbl}
\definecolor{GrayTableHighlight}{gray}{0.9}
\usepackage{hyperref}
\hypersetup{
	hidelinks
}
\usepackage[capitalize]{cleveref}
\crefname{line}{line}{lines}
\usepackage{setspace}
\usepackage[bordercolor=white,backgroundcolor=gray!30,linecolor=black,colorinlistoftodos]{todonotes}

\usepackage{xcolor}
\usepackage[framemethod=tikz]{mdframed}
\mdfsetup{skipabove=0pt,skipbelow=0pt}
\mdfdefinestyle{highlight}{hidealllines=true,backgroundcolor=blue!20,innerleftmargin=3pt,innerrightmargin=3pt,leftmargin=-3pt,rightmargin=-3pt,outerlinewidth=0pt,innertopmargin=0pt, innerbottommargin=0pt}

\usepackage{soul}
\sethlcolor{yellow}

\usepackage[normalem]{ulem}

\newcommand{\mmessage}[1]{\normalfont\ensuremath{\langle #1 \rangle}}
\newcommand{\smessage}[2]{\normalfont\ensuremath{\langle #1 \rangle_{\sigma_{#2}}}}
\newcommand{\tsmessage}[3]{\textnormal{\ensuremath{\langle \textsf{#1}, #2 \rangle_{\sigma_{#3}}}}}
\newcommand{\tnsmessage}[2]{\textnormal{\ensuremath{\langle \textsf{#1}, #2 \rangle}}}

\newcommand{\kec}[1]{\ensuremath{\mathtt{KEC}(#1)}}
\newcommand{\cs}[1]{\ensuremath{\textit{CS}(#1)}}

\newcommand{\quorum}[1]{\ensuremath{\textit{Quorum}(#1)}}
\newcommand{\quorumoptimal}{\ensuremath{\textit{Quorum}_{\textit{opt}}}}
\newcommand{\fb}{\ensuremath{\mathit{FB}}}

\newcommand{\lock}[1][]{\emph{lock#1}}

\newcommand{\unlock}[1][]{\emph{unlock#1}}

\newcommand{\ibfp}[1]{IBFT-block-finalisation-protocol}
\newcommand{\ibftmone}[1]{IBFT-M1}
\newcommand{\ibftmonep}[1]{\ibftmone{} protocol}
\newcommand{\ibfpmone}[1]{IBFT-M1-block-finalisation-protocol}
\newcommand{\ibftmtwo}[1]{IBFT-M2}
\newcommand{\ibftmtwop}[1]{\ibftmtwo{} protocol}

\newcommand{\devptwop}[1]{\textsf{\DH}\ensuremath{\mathsf{ \Xi Vp2p}}}
\renewcommand{\doublespacing}{{}}

\newtheorem{theorem}{Theorem}
\newtheorem{lemma}{Lemma}

\numberwithin{corollary}{lemma}

\theoremstyle{definition}
\newtheorem{definition}{Definition}

\begin{document}
\title{Correctness Analysis of Istanbul Byzantine Fault Tolerance}
		
\author{Roberto Saltini \and David Hyland-Wood 	}

\date{\textit{PegaSys (ConsenSys)} \\
	\texttt{\{name\}.\{surname\}@consensys.net}\\
		\today}

\doublespacing

\maketitle 

\begin{abstract}
Istanbul Byzantine Fault Tolerance (IBFT) is a proof-of-authority (PoA) blockchain consensus protocol that ensures immediate finality. A formal model of IBFT is presented, and the correctness of IBFT is analysed. We show that the protocol does not guarantee Byzantine-fault-tolerant persistence when operating on an eventually synchronous network. We also show that the protocol does not guarantee Byzantine-fault-tolerant liveness when operating on an eventually synchronous network. Modifications are proposed to the protocol to ensure both Byzantine-fault-tolerant persistence and liveness in eventually synchronous settings.
\end{abstract}

%
%

\section{Introduction}

\subsection{Background}

Blockchains, also known as \emph{distributed transaction ledgers}, are immutable, append-only data structures that maintain an ordered set of transactions organised into blocks whose order is established cryptographically. Such structures are therefore \emph{tamper evident} in that modification of either a block or the order of blocks may be readily determined at runtime. A blockchain starts with an initial block called the \emph{genesis block}. The first widely adopted blockchains were Bitcoin \cite{Nakamoto_bitcoin} and Ethereum \cite{yellowpaper}.

The Ethereum blockchain augments the standard value transfer capability of each transaction with the possibility to specify instructions of a Turing-complete language to execute on a sandboxed runtime. 
The runtime, called the Ethereum Virtual Machine (EVM) modifies the Ethereum global state maintained by each node \cite{yellowpaper}. 
This means that any user of the Ethereum blockchain has the capability to create decentralised applications, called \emph{smart contracts}, that can govern the interaction between the different users of the system.
One of the first use cases for Ethereum was the creation of escrow smart contracts eliminating the need for a trusted 3rd party.

Within the blockchain context, the role of consensus protocols is to define algorithms by which all nodes on the network can agree on the canonical block sequence. 
Most of these protocols have some level of Byzantine-fault-tolerance (BFT) \cite{Lamport:1982:BGP:357172.357176} which means that the protocols can cope with some fraction of the nodes, called Byzantine nodes, being faulty and behaving in any arbitrary way, including colluding and delaying message delivery to honest nodes.
The standard Ethereum consensus protocol is designed to operate in a permissionless setting where any node can join or leave the network at any point in time and can propose blocks to be added to the blockchain. 
Ethereum requires each node to spend compute effort in resolving a hard cryptographic puzzle before it can propose a block. 
This type of technique, called proof of work (PoW) \cite{Dwork:1992:PVP:646757.705669,Nakamoto_bitcoin}, is used to prevent Sybil attacks where a node is able to gain power in a system by creating multiple pseudonymous identities \cite{Douceur2002}.
Proof of stake (PoS) \cite{ppcoin-pos}, proof of space (PoSpace) \cite{park2015spacemint} and proof of elapsed time (PoET) \cite{poet-security} are examples of alternative techniques for preventing Sybil attacks in a permissionless context that have been proposed over the years.
The Ethereum PoW consensus protocol only guarantees probabilistic immutability of the blockchain where the probability that a block will never be removed from the blockchain increases as new blocks are added to the blockchain.

The IBFT consensus protocol was developed around early 2017 by AMIS \cite{eip650}, and has been implemented in Quorum \cite{jpmorgangithub}.
IBFT was created to provide an alternative consensus protocol for the Ethereum blockchain that was better suited for either private or consortium blockchains, where deterministic immutability of the blockchain is often a requirement and spending significant computational effort is less desirable.
Specifically, IBFT guarantees immediate finality, i.e. each block added to the blockchain will never be removed from it (except if the number of Byzantine actors exceed a given threshold), and uses proof of authority (PoA) as a Sybil attack prevention technique.
In PoA systems, nodes are not required to spend compute effort to solve a cryptographic puzzle as only a set of special nodes, called \emph{validators}, are authorised to propose new blocks.

IBFT specifies minimal modifications to the standard Ethereum block header structure and block validation ruleset while maintaining all other Ethereum capabilities (e.g. smart contracts). An IBFT blockchain maintained by each node is constituted of \emph{finalised blocks} where each finalised block contains an Ethereum block and a finalisation proof used to attest that consensus was reached on the Ethereum block included in the finalised block.

IBFT adapts the state machine replication protocol, Practical BFT (PBFT) \cite{Castro:1999:PBF:296806.296824}, to work in a blockchain setting and integrates it with the dynamic validator set voting mechanism originally designed for Clique \cite{eip225} that allows validators to be added to, or removed from, the set.
This makes IBFT well suited for consortium or private blockchains where there are limitations on which nodes can join the network and the permissioning level of each of the nodes.
Private blockchains are defined by allowing only nodes from a single organisation to join a peer-to-peer network.
A consortium network, by contrast, is defined by a group of organisations that form a consortium, thus nodes from any member of the consortium can join such a network. In some consortium networks read access is made available to the public.
In both private and consortium networks block creation is limited to a subset of nodes.

To date, Kaleido, a blockchain project that offers blockchain-as-a-service, has deployed the IBFT protocol on approximately 470  blockchain networks\footnote{Data provided by Kaleido in a private conversation.}.
Given that Kaleido is only one of the 19 projects offering blockchain-as-a-service \cite{19-bas}, it is likely that the actual number of blockchain networks currently using IBFT far exceeds this number.



\subsection{Related Work}\label{sec:related-work}

%
%
While IBFT is inspired by PBFT \cite{Castro:1999:PBF:296806.296824}, there are a few significant differences between the two protocols:
\begin{itemize}
	\item in IBFT there is no client submitting requests to the network, instead all of the validators can in turn propose a block to the network of validators;  
	\item IBFT allows for two types of nodes: validators that take part in the consensus protocol and standard nodes that validate blocks but do not participate in the consensus protocol; 
	\item the set of validators in PBFT is static whereas IBFT features a dynamic validator-set where validators can be added to, or removed from, the set;
	\item IBFT specifies a simplified version of the so called View-Change message of PBFT and does not include the so called New-View message included in the PBFT protocol;
	\item while IBFT does not use checkpoints explicitly, each IBFT block can be considered the IBFT equivalent of a PBFT checkpoint.
\end{itemize}

PoA consensus protocols such as Clique \cite{eip225}, and two Aura protocols that share the same name but are different protocols \cite{the-attack-of-the-clones,aura,aura2}, do not guarantee immediate finality, i.e. there is no deterministic guarantee that the blockchain history will not be rewritten. 
The key difference between Clique and the two Aura protocols is the fork choice rule that decides which is the preferred chain when there are two competing chains of blocks.
Also, Ekparinya et al. \cite{the-attack-of-the-clones} demonstrated that these protocols provide only limited tolerance to Byzantine actors.

Honey Badger \cite{honeybadger} is an example of a PoA blockchain consensus protocol that ensures immediate finality, but relies on a probabilistic common coin protocol to ensure liveness. 
Therefore, the addition of new blocks to the local blockchain of honest nodes is a probability function.
This makes Honey Badger better suited to asynchronous networks than networks that can guarantee partial synchrony.


DBFT (Democratic BFT)  \cite{dbft} is an example of a PoA blockchain consensus protocol that ensures optimal Byzantine-fault-tolerant resilience under the assumption that the network is reliable, i.e. all messages that are sent are always eventually delivered.
For the best case scenario, DBFT requires 4 message delays to decide on the creation of a new block.
Compared to other BFT protocols, that usually rely on a coordinator or block proposer to ensure liveness, DBFT uses the notion of a weak-coordinator where nodes are not necessarily required to wait for the message sent by the coordinator in order to proceed. 
This is advantageous when the coordinator network is likely to experience delays in the delivery of messages to the other nodes.

Algorand \cite{algorandagreement} is another example of a PoA blockchain consensus protocol that ensures optimal Byzantine-fault-tolerant persistence and relies on the assumption that the network is reliable.
An important contribution of Algorand \cite{algorand} is the definition of ``secret cryptography sortition'' where the block proposer is not revealed to the network until the block is broadcast. Therefore, the risk of an adversarial targeted attack on the proposer is reduced.


%

\subsection{Contributions}
This work presents two major contributions.
\begin{enumerate}
	\item 
		We improve upon the IBFT protocol definition provided by Saltini \cite{IBFT-Liveness-Analysis} by using a more precise semantics and better specifying the process structure for each node by means of process-algebra-like notation \cite{CSP,CSS,ACP}. 
		We provide a formal specifications of the protocol in process algebra-like syntax. 
		The immediate advantage of this notation is an unambiguous definition of the protocol which enables us to perform a rigorous analysis of the  robustness of the protocol.
		This will also allow future refinement works and it provides a precise specification as starting point for other researchers who want to explore other aspect of the protocol such as performances.

	\item We also improve upon the analysis of IBFT protocol performed by Saltini \cite{IBFT-Liveness-Analysis} which is limited to analysing the liveness property only, leaving the persistence analysis for future work.
	
	In this paper we intend to address this knowledge gap by performing the first robustness analysis of the IBFT protocol when operating in an eventually synchronous network \cite{Dwork:1988:CPP:42282.42283}.
	As a result of this analysis, we show that the IBFT protocol does not guarantee Byzantine-fault-tolerance persistence and liveness when operating in an eventually synchronous network.
	The results presented here provide an opportunity to put additional measures in place for old and new deployments of IBFT networks so that the likelihood of persistence and liveness issues is reduced.
	An example of such a measure is deploying nodes on a network that can guarantee a high level of synchrony.

	\item As part of this work we describe in detail how the IBFT protocol can be modified to guarantee optimal Byzantine-fault-tolerant persistence. We also describe two high-level solutions for addressing the lack of Byzantine-fault-tolerant liveness.
	The performance of each solution is dependent on each specific use case of the blockchain network where the protocol is deployed. 
	This could potentially lead to the development of a configurable protocol so that users can choose which of the two solutions is best suited to their specific needs.
	 
	Moreover, one of the two solutions to the liveness limitation (the PBFT-like solution) leads to a blockchain consensus protocol that provides optimal Byzantine-fault-tolerance persistence and liveness even when operating on networks that may experience an initial period of asynchrony, where messages sent during this period may be lost.
	This compares to the many blockchain consensus protocols currently available (see \cref{sec:related-work}) that guarantee liveness only for networks where the eventual delivery of all messages is guaranteed.
	Furthermore, we argue that the PBFT-like solution also allows for an additional optimisation to reduce the minimum number of communication phases from 3 down to 2.
\end{enumerate}

\subsection{Paper Outline}
The paper is organised as follows: 
In Section \ref{sec:system-model} we present our analysis model, and define robustness for the IBFT protocol as the logical conjunction of two properties: persistence and liveness. Persistence guarantees blockchain consistency and immutability amongst all honest nodes, while liveness guarantees that transactions submitted to the system will eventually be included in the blockchain (distributed transaction ledger). 
In \cref{sec:protocol-description} we  describe the IBFT protocol. 
\Cref{sec:definitions} introduces a series of definitions that will be used in the following analysis of the IBFT protocol.
In Section~\ref{sec:safety-analysis}
we present the persistence analysis of the IBFT protocol which shows that  IBFT does not guarantee Byzantine fault tolerance when operating in an eventually synchronous network. 
In \cref{sec:modification-ibftmonep-achieve-optimal-byzantine-fault-tolerance-persistence}, we then describe and analyse modifications to the IBFT protocol that will ensure optimal Byzantine fault tolerance under eventual synchrony.
Section~\ref{sec:liveness-analysis}
discusses the liveness property of the protocol showing that IBFT is not live in eventually synchronous networks, and explores two potential modifications to the IBFT protocol that will ensure liveness under these conditions.

\section{System Model}\label{sec:system-model}
\paragraph{Asynchronous nodes}
We consider a system composed of an unbounded number of asynchronous nodes, each of them maintaining a local copy of the blockchain obtained by adding blocks to it as specified by the IBFT protocol.
We assume that all nodes have the same genesis block.

\paragraph{Network Model}\label{par:network-model}
The IBFT protocol relies on the Ethereum \devptwop{} protocol for the delivery of all protocol messages.
We model the gossip network as an eventually synchronous network, also called a partially synchronous network, as defined in Dwork et al \cite{Dwork:1988:CPP:42282.42283}, where there exists a point in time called global stabilisation time (GST), after which the message delay is bounded by a constant, $\Delta$.
Before GST there is no bound on the message delay.
In other words, we assume that before GST messages may be lost and never delivered, even after GST.
However, once GST is reached, messages are guaranteed to be delivered within $ \Delta $. The value of $ \Delta $ is however unknown.

\paragraph{Failure Model}
We consider a Byzantine failure mode system, where Byzantine nodes can behave arbitrarily. In contrast, honest nodes never diverge from the protocol definition.
We denote the maximum number of Byzantine nodes that an eventually synchronous network of $n$ nodes can be tolerant to with $f(n)$.
As proven in Dwork et al \cite{Dwork:1988:CPP:42282.42283}, the relationship between the total number of nodes, $n$, and the maximum number of Byzantine nodes can be expressed as follows:
\begin{equation}
	f(n) \equiv \left \lfloor \frac{n-1}{3} \right \rfloor
\end{equation}
\doublespacing
\paragraph{Cryptographic Primitives}
The IBFT protocol uses the Keccak hash function variant as per the Ethereum Yellow Paper \cite{yellowpaper} to produce digests of blocks.
We assume that the Keccak hash function is collision-resistant.\\
The IBFT protocol uses the standard Ethereum ECDSA signing algorithm \cite{yellowpaper}.
We assume that such a digital signature scheme ensures uniqueness and unforgeability.
Uniqueness means that the signatures generated for two different messages are also different with high probability.
The unforgeability property ensures that Byzantine nodes, even if they collude, cannot forge digital signatures produced by honest nodes.\\
We use \smessage{m}{v} to denote a message $m$ signed by validator $v$.

\paragraph{IBFT Robustness Property}
For the purpose of defining the robustness property, we consider the IBFT protocol as implementing a \emph{distributed permissioned transaction ledger} with immediate finality.\\
A distributed transaction ledger maintains an append-only fully-ordered set of transactions and ensures its consistency amongst all honest nodes that participate in the protocol.\\
Each node maintains a local copy of the transaction ledger organised as a chain of blocks, or blockchain.
Our definition of robustness for the IBFT protocol is based on the definition of robustness for public transaction ledgers provided in Garay et al \cite{GKL:2018}.\\
For the purpose of this definition, the position of a transaction within the transaction ledger implemented by the IBFT protocol is defined as a pair with the first component corresponding to the height of the block including the transaction and the second component corresponding to the position of the transaction within the block.
\begin{definition}\label{def:robustnes}
	The IBFT protocol implements a robust distributed permissioned transaction ledger with immediate finality and $t$-Byzantine-fault-tolerance  if, provided that no more than $t$ validators are Byzantine, it guarantees the following two properties:
	\begin{itemize}
		\item \textbf{Persistence.} If an honest node adds transaction $ T $ in position $ i $ of its local transaction ledger, then (i) $ T $ is the only transaction that can ever be added in position $i$ by any other honest node, (ii) $T$ will eventually be added to the local transaction ledger of any other honest node.
		
		\item \textbf{Liveness.} Provided that a transaction is submitted to all honest validators, then the transaction will eventually be included in the distributed permissioned transaction ledger.
	\end{itemize}
\end{definition}

\SetKwBlock{Init}{Initialisation:}{}
\SetKwBlock{Functions}{Functions:}{}
\SetKwBlock{Expansions}{Macro Expansions:}{}
\SetKwBlock{Procedures}{Procedures:}{}
\SetKwBlock{UponRules}{Guarded Commands:}{}
\SetKwData{Reception}{reception of}
\SetKwData{HasReceived}{has received}
\SetKwFunction{Proposer}{proposer}
\SetKwFunction{AV}{validators}
\SetKwFunction{Valid}{isValidBlock}
\SetKwFunction{N}{n}
\SetKwFunction{SizeOf}{sizeOf}
\SetKwFunction{NumOf}{numOfReceived}
\SetKwFunction{Received}{received}
\SetKw{Broadcast}{broadcast}
\SetKw{Multicast}{multicast}
\newcommand{\Multicastv}[1]{\ensuremath{\Multicast_v\;#1}}
\SetKwFunction{WellFormedToAddFinalProof}{wellFormedToAddFinalProof}
\SetKwFunction{RoundTimerTimeout}{roundTimerTimeout}
\SetKwFunction{CreateNewProposedBlock}{createNewProposedBlock}
\SetKwFunction{BlockHeight}{blockHeight}
\SetKwFunction{ExtractBlock}{block}
\SetKwFunction{ExtractFinalisationProof}{finalisationProof}
\SetKwFunction{EthAddressRecover}{recValidator}
\SetKwFunction{TimeoutForRoundZero}{timeoutForRoundZero}
\SetKwFunction{k}{K}
\SetKwFunction{Sign}{sign}
\SetKw{From}{\:from\:}
\SetKw{Stop}{stop}
\SetKw{Start}{start}
\SetKw{Set}{set}
\SetKw{Expiry}{expiry of}
\SetKw{And}{\:and\:}
\SetKw{Or}{\:or\:}
\SetKw{In}{\:in\:}
\SetKw{Not}{\:not\:}
\SetKw{With}{\:with\:}
\SetKw{Let}{let}
\SetKw{To}{to}
\SetKwFor{Upon}{\hspace{-1em}$ c_1 \equiv $ upon}{do}{}
\SetKwProg{Fn}{}{$\equiv$}{}
\SetKwProg{Gc}{}{$\rightarrow$}{}
\SetKwProg{Gcp}{\hspace{0.4em}$ + $}{$\rightarrow$}{}
\SetKwProg{Proc}{def}{:}{}
\SetKwIF{Expand}{Expand2}{Expand3}{}{expands to:}{}{}{}
\SetKwIF{EmptyIf}{}{}{}{}{}{}{}
\newcommand{\avibfp}[1][]{\AV{\ensuremath{{chain}_v[0:h_{#1}-1]}}}
\newcommand{\nmacro}[1][]{\N{\ensuremath{{chain}_v[0:h_{#1}-1]}}}
\newcommand{\nhmacro}[1][]{\ensuremath{n_{h,v}}}
\newcommand{\avhmacro}[1][]{\ensuremath{{\AV}_{h,v}}}

\newcommand{\proposermacro}[1][]{\Proposer{\ensuremath{{chain}_v[0:h-1],  r_{h,v}}}}
\newcommand{\hproposer}[1][]{\ensuremath{{\Proposer}_{h,v}(r_{h,v})}}
\LinesNumbered

\SetAlgoNoLine
\SetAlgoHangIndent{2em}
\SetNlSty{textbf}{}{:}
\DontPrintSemicolon
\SetAlFnt{\small}
\section{IBFT Protocol Specification}\label{sec:protocol-description}
The presentation of the IBFT protocol provided in this section is based on EIP 650~\cite{eip650} and the actual implementation available on GitHub \cite{jpmorgangithub}.

\subsection{Overview}
In alignment with the blockchain literature, we define the height of a finalised block $\fb$ as the distance in finalised blocks between $\fb$ and the genesis block that has height 0. 
The first finalised block after the genesis block has height 1, the next one has height 2 and so on.

For each block height, the IBFT consensus protocol can be considered running sequential different instances of what we call the \emph{\ibfp{}}.
The aim of the $h$-th instance of the \ibfp{} is to decide on the block to be added at height $h$, generate the related finalised block and broadcast it to all nodes. 
Only a subset of the entire set of IBFT nodes can participate in the $h$-th instance of the block finalisation protocol. 
We call this set of nodes the \emph{validators for height/instance $h$} and refer to each member of this set as a \emph{validator for height/instance $h$}.
We also refer to all of the nodes not included in the validator set for height/instance $h$ as \emph{standard nodes}.
We often omit \emph{for height/instance $h$} when this is clear from the context.
The set of validators for each instance $h$ of the \ibfp{} is deterministically computed as a function of the chain of blocks from the genesis block until the block with height $h-1$.

Once the \ibfp{} for an honest validator reaches a decision on the block to be included in the finalisation block with  height $h$, it creates a finalised block by adding the  finalisation proof  to the block and propagates the finalised block to all other nodes in the network, both validators and standard nodes.
Finalised blocks are transmitted using the standard Ethereum block synching protocol by which each validator (i) transmits a newly created finalised block using the \devptwop{} gossip protocol \cite{wireprotocol} and (ii) asks its peers for the availability of new finalised blocks either when the validator starts or when the validator receives a finalised block with a height higher than the next expected finalised block height.
The purpose of the finalisation proof is to allow any node, even nodes that did not participate in the \ibfp{}, to verify that a decision on the block inclusion in the blockchain at height $h$ was reached by a correct execution of the \ibfp{} despite the presence of a potential number of Byzantine validators.
As described in \cref{sec:ibfp}, the finalisation proof is composed of signatures over the Ethereum block, called \emph{commit seals}, that are sent by validators as part of the Commit messages exchanged during the \ibfp{}.

\subsection{System Specification}
Each node $ v $ runs the $ \textit{IBFT}(v) $ protocol as defined by the following process-algebra-like \cite{CSP,CSS,ACP} recursive equations
\begin{align}
&IBFT(v) \equiv ibftInit(v) \cdot ibftExecution(v) \label{eqn:IBFT-v}\\
&ibftExecution(v) \equiv  \big( igc_1(v)  + igc_2(v) \big) \cdot ibftExecution(v)\label{eqn:ibftExecution-v}
\end{align}
where $ \cdot $ indicates sequential composition, $ + $ indicates non-deterministic choice, and $ ibftInit(v) $ and $ \big(igc_i(v)\big)_{i\in\{1,2\}} $ correspond to sequential atomic steps defined by the pseudocode in \Cref{algo:ibft-protocol}.

\begin{sloppypar}
When either \cref{ln:start-ibfp-1} or \cref{ln:start-ibfp-2} of \Cref{algo:ibft-protocol} is executed, node $ v $ starts a new $\textit{\ibfp{}}(h_v,v)$ process which is run in parallel to the $ \textit{IBFT}(v) $ process.
The $\textit{\ibfp{}}(h,v)$ process is described by the following process-algebra-like equations in which we use the same notation used in the definition of $ \textit{IBFT}(v)$
\end{sloppypar}
\begin{align}
&\textit{\ibfp{}}(h,v) \equiv fpInit(h,v) \cdot fpExecution(h,v)\label{eqn:ibfp-h-v}\\
&fpExecution(h,v) \equiv \Big( \Plus\limits_{i=1}^7 fpgc_i(h,v) \Big) \cdot fpExecution(h,v)\label{eqn:fpExecution-v-h}
\end{align}
where $ \textit{fpInit}(h,v) $ and  $ \big(\textit{fpgc}_i(h_v)\big)_{i \in \{1,\ldots,7\}} $ are atomic sequential steps defined in \Crefrange{algo:ibft-block-finalisation-protocol-start-and-init}{algo:ibft-block-finalisation-protocol-commit-and-round-change}.

The process  $\textit{\ibfp{}}(h_v,v)$ is stopped when \cref{ln:stop-ibfp-2} of \Cref{algo:ibft-protocol} is executed.

Starting process $\textit{\ibfp{}}(h,v)$ corresponds to starting the $ h$-th instance of the \ibfp{} for node $ v $.

\subsection{Notation}
The notation used in the pseudocode in \crefrange{algo:ibft-protocol}{algo:ibft-block-finalisation-protocol-commit-and-round-change} is as follows:
\begin{itemize}
	\item $ c \rightarrow S $ corresponds to a guarded command where the code $ S $ can be executed only if condition $ c $ is true;
	\item The scope of variables used in the guards of guarded commands extends to the pseudocode specifying the command;
	\item If the domain of existential quantifiers on variables is not specified, then it is assumed to correspond to the expected domain of the function argument, message or tuple component where the  variables are used;
	\item For brevity of notation, existential quantifiers applied to messages or tuples correspond to the application of existential quantifiers to each of the elements of the message or tuple that is not already assigned to a bounded variable. For example, $ \exists \mmessage{x,y,h_v, z} \in V $, where $ h_v $ is the only bounded variable already defined elsewhere, is an abbreviation for $ \exists x,y,z: \mmessage{x,y,h_v,z} \in V $;
	\item All functions in \FuncSty{typewriter font} are defined in the remainder of this section, whereas all functions in \textit{italic font} are defined in the pseudocode;
	\item $ \{ F(m) \mid m \in V \land P(m) \} $ corresponds to the set obtained by applying the function $ F $ to all the elements of $ V $ for which predicate $P$ is true; when $ F(\cdot) $ corresponds to the projection onto the $ j $-th coordinate $ \pi_j(\cdot) $, then we omit it, e.g. we write $ \{ x \mid (x,y) \in V \} $ rather than $ \{\pi_1(\mathbf{t}) \mid \mathbf{t} \in V \} $;	
	\item $ \Multicastv{m}\;\To\;V $ corresponds to node $ v $ multicasting message $ m $ to all nodes in set $ V $;
	The delivery of messages to the recipients happens according to the eventually synchronous network model as detailed by the Network Model definition in \cref{sec:system-model};
	\item When a node $ v  $ receives a message $ m $, it assigns a unique identifier $ id $ to it and places the tuple $ (m,id) $ into the set $ IRM_v $;
	\item $ {RM}_v $ corresponds to the set of all messages received by node $v$ without the identification value;
	In other words, $ {RM}_v $ is the projection of  $ \mathit{IRM}_v $ onto the first coordinate, i.e. $ RM_v \equiv \{ m \mid (m,id)\in TRM_v \} $
	\item $ \mathcal{P}({RM}_v) $ corresponds to the power set of $ {RM}_v $;
	\item For clarity of notation, we use \SizeOf{$ M $} to indicate the size of the set $M$, i.e $\SizeOf{\ensuremath{M}} \equiv \|M\|$;	
	\item The symbol $*$ denotes any value;
	\item $\BlockHeight{\fb}$ is defined as the height of the finalised block $\fb$;
	\item Each validator $v$ stores its local blockchain in ${chain}_v$;
	\item ${chain}_v[n]$ corresponds to the finalised block with height $n$, while ${chain}_v[n:m]$ corresponds to a sub-chain including all of the finalised blocks from height $n$ to height $m$;
	\item  \AV{${chain}_v[0:h-1]$} represents the set of authorised validators for instance $h$ of the \ibfp{}.
			The definition of the \AV{$\cdot$} function is not presented here as it is outside the scope of this work and does not have any relevance to the results presented here.
			For the same reason, we do not describe the protocol that can be used to add or remove validators to/from the validator set of each instance of the \ibfp{};
	\item  $\N{\ensuremath{{chain}_v[0:h-1]}}$ represents the number of validators for instance $h$ of the \ibfp{}, i.e. $\N{\ensuremath{{chain}_v[0:h-1]}} \equiv \SizeOf{\AV{\ensuremath{{chain}_v[0:h-1]}}} $;
	\item \ExtractFinalisationProof{\fb}  denotes the finalisation proof included in the finalised block \fb{};
	\item \WellFormedToAddFinalProof{$ B $} is defined as true  if block $ B $ allows adding the finalisation proof to it; it is defined to be false  otherwise;
	\item \ExtractBlock{\fb} denotes the block included in the finalised block \fb{};
	\item The function $\Valid{\ensuremath{B,B_{parent}}}$ is defined to be true if and only if block $B$ is a valid Ethereum block with parent $B_{parent}$. 
			For the purpose of this work, we consider that \linebreak $\Valid{\ensuremath{B,B_{parent}}}$ only verifies the following fields of the standard Ethereum header: parentHash, stateRoot, transactionsRoot, receiptsRoot, logsBloom, number, gasLimit, gasUsed;
	\item \kec{B} corresponds to the application of the Keccak hash function to $ B $;
	\item \Sign{$ H,v $} corresponds to the signature of the validator $v$ over the hash $H$;
	\item \EthAddressRecover{\ensuremath{H,signature}} corresponds to the validator whose signature of the hash $H$ corresponds to $signature$, i.e. $ \EthAddressRecover{\ensuremath{H,\Sign{ \ensuremath{H,v }}}} \equiv v $ ;
	\item $\quorum{n}$ is defined as $ f(n) \cdot 2 + 1 $;
	\item $ now(v) $ corresponds to the current time for node $ v $;
	\item 
		\Start $ p(\cdot) $ corresponds to spawning a new process $ p(\cdot) $.
		This can be expressed in process algebra by replacing the current process with the parallel composition of $ p(\cdot) $ with the current process.
	\item 
		\Stop $ p(\cdot) $ corresponds to stopping process $ p(\cdot) $.
		This can be expressed in process algebra by means of either communication channels where the process executing the \Stop statement sends a stop signal to the target process or via the used of shared state variables.
\end{itemize}

\begin{algorithm}[t!]
	\Functions{ 
		\lFn{\mbox{\hspace{-0.4em}}\quorum{n}}{$ f(n) \cdot 2 + 1 $}
		\Fn{${isValidFinalisationProof}(\fb,v)$}{
			\SizeOf{$ \{
				cs \in \ExtractFinalisationProof{\fb} \mid \newline$
				\mbox{\hspace{3em}}$\EthAddressRecover{\ensuremath{\kec{\ExtractBlock{\fb}},cs }}\in \avibfp[v]\}$\newline
			    \mbox{\hspace{1.3em}}} $\geq \quorum{\nmacro[v]}$
		}	
		\Fn{${isValidFinalisedBlock}(\fb, v)$\label{ln:check-block-validity}}{
			${isValidFinalisationProof}(\fb,v) \:\land$\newline
			\mbox{\hspace{-2em}}\Valid{${\ExtractBlock{\fb},}$
			$\ExtractBlock{\ensuremath{{chain}_v[h_v-1]}}$}
			
	}}
	\Init{
		\Fn{$ ibftInit(v)  $}
		{
		${chain}_v[0] \gets \textit{genesis block}$\;
		$h_v \gets 1 $\;
		$ oldPrePrepareMessagesAlreadyProcessed_v \gets \{\} $\;
		\If{$ v \in \avibfp[v] $}
		{
			\Start $\textit{\ibfp{}}(h_v,v)$\label{ln:start-ibfp-1}
		}
	}
	}
	\UponRules{	 
		\Gc{\label{ln:finalise-block-start}%
			$igc_1(v) \equiv \newline \big( {\ensuremath{\exists \tnsmessage{FINALISED-BLOCK}{\fb} \in {RM}_v }} :$\newline
		    \mbox{\hspace{1em}}\mbox{$\BlockHeight{\ExtractBlock{\fb}} = h_v \:\land$}\newline
		    \mbox{\hspace{1em}}${isValidFinalisedBlock}(\fb,v)\newline\big)$
	        } {	
	        	\Indp\Indp        
				${chain}_v[h_v] \gets \fb$\label{ln:add-block}\;
				\If{$ v \in \avibfp[v] $}{
					\Stop $ \textit{\ibfp{}}(h_v,v)$\label{ln:stop-ibfp-2}
				}
				$h_v \gets h_v + 1$\label{ln:increase-h-v}\;
				\If{$ v \in \avibfp[v] $}
				{
					\Start $\textit{\ibfp{}}(h_v,v)$\label{ln:start-ibfp-2}
				}
	}
		\Gc{$igc_2(v) \equiv \newline
		\big(\exists (\tsmessage{PRE-PREPARE}{h_{pp}, r_{pp}, B}{sv },id) \in \mathit{IRM}_v: \newline
		\mbox{\hspace{3em}}id \notin oldPrePrepareMessagesAlreadyProcessed_v \land \newline 		
		\mbox{\hspace{3em}}h_{pp} < h_v \; \land \newline 
		\mbox{\hspace{3em}}  sv = \Proposer{\ensuremath{{chain}_v[0:h_{pp}-1],  r_{pp}}} \land \newline
		\mbox{\hspace{3em}} B = \ExtractBlock{\ensuremath{{chain}_v[h_{pp}]}}\newline\big)$\label{ln:send-commit-for-old-blocks}}
	{
		\Indp\Indp
		\Multicastv{\tsmessage{COMMIT}{h_{pp}, r_{pp}, \kec{B},\cs{B,v}}{v}} \To $ \AV_{h_{pp},v} $\;
		$ oldPrePrepareMessagesAlreadyProcessed_v \gets oldPrePrepareMessagesAlreadyProcessed_v \cup \{id\} $\;
	}

	}	
	\caption{Sequential steps of the $ \textit{IBFT}(v) $ process.}
	\label{algo:ibft-protocol}
\end{algorithm}

\subsection{$ \textit{IBFT}(v) $ Specification}

As described by \Cref{eqn:IBFT-v,eqn:ibftExecution-v}, the IBFT protocol executed by a node $ v $ corresponds to executing the initialisation procedure $ ibftInit(v) $ followed by a recursive execution of the non-deterministic choice between guarded commands $ igc_1(v) $ and $ igc_2(v) $.

Guarded command $ igc_1(v) $ is enabled when node $ v $ receives a valid finalised block for the next expected block height, i.e. $ h_v $.
We say that a finalised block \fb{} is valid for validator $v$ if and only if ${isValidFinalisedBlock}(\fb,v)$,
which  corresponds to the logical conjunction of the following two conditions, is true:
\begin{itemize}
	\item at least  $\quorum{n} $ of the commit seals included in the finalisation proof of \fb{} are signed by validators included in the set of validators for the $h$-th instance of the \ibfp{}, where $n \equiv \N{\ensuremath{{chain}_v[0:h-1]}}$;
	\item the block included in \fb{} is a valid Ethereum block.
\end{itemize}

When guarded command $ igc_1(v) $ is executed,  then  (i) $v$ adds the finalised block to its local blockchain, (ii) if $v$ is a validator for the $h_v$-th instance of the \ibfp{} then $v$ aborts the $h_v$-th instance of the \ibfp{}, (iii) $v$ advances $h_v$ to $h_v+1$ and (iv) if $v$ is a validator for the \ibfp{} instance for the new value of $h_v$, then $v$ starts that instance.

Guarded command $ igc_2(v) $ will be described in the following sub-section.

\subsection{$ \textit{\ibfp{}}(h,v)$ Specification}\label{sec:ibfp}

\begin{algorithm}[t!]
	\caption{Sequential steps of the $ \textit{\ibfp{}}(h,v) $ process (initialisation).}
	\label{algo:ibft-block-finalisation-protocol-start-and-init}
	\Expansions{
		\lExpand{\nhmacro[]}{\nmacro[]\label{ln:nhmacro}}
		\lExpand{\avhmacro[]}{\avibfp[]\label{ln:avibfp}}
		\lExpand{\hproposer[] } {\proposermacro[]\label{ln:hproposer}}}
	\Functions{
		\lFn{$ {RoundTimerTimeout}(r)$}{$\TimeoutForRoundZero \cdot 2^r$}	
		\lFn{$ \cs{B,v} $}{\Sign{$ \kec{B},v $}}
	}
	\Init{
		\Fn{$ fpInit(h,v) $}
		{
		${lockedBlock}_{h,v} \gets \bot $\;\label{ln:unlock-at-initialisation}
		$ StartNewRound(0,h,v) $\;
	   }
	}
	\Procedures{
		\Proc{$moveToNewRound(r, h, v)$}
		{
			$r_{h,v} \gets r$\;
			${roundAlreadyStarted}_{h,v} \gets {false}$\;
			${acceptedBlock}_{h,v} \gets \bot$\;\label{algo2:set-accept-pre-prepare-1}

		}
		\Proc{$ StartNewRound(r,h,v) $\label{ln:start-new-round}}{
			$moveToNewRound(r, h, v)$\;
			${roundAlreadyStarted}_{h,v} \gets {true}$\;
			${commitSent}_{h,v} \gets {false}$\;\label{ln:reset-commit-sent}
			${finalisedBlockSent}_{h,v} \gets false$\;\label{ln:reset-finalised-block-sent}
			${roundTimerExpiration}_{h,v}[r_{h,v}] \gets now(v) +  RoundTimerTimeout(\ensuremath{r_{h,v}})$\;\label{ln:start-timer}
			\If{$v = \hproposer[] $}{
				\eIf{${lockedBlock}_{h,v} \neq \bot$}{
					$B \gets {lockedBlock}_{h,v} $\;
				}
				{
					$B \gets \CreateNewProposedBlock{\ensuremath{h,v}} $
				}
				\Multicastv{\tsmessage{PRE-PREPARE}{h, r_{h,v}, B}{v}} \To \avhmacro[]\;
			}
		}
		\Proc{$moveToNewRoundAndSendRoundChange(r, h, v)$}
		{
			$moveToNewRound(r, h, v)$\;
			\Multicastv{\tsmessage{ROUND-CHANGE}{h, r_{h,v}}{v}} \To \avhmacro[]\;
		}
	}
\end{algorithm}

\begin{algorithm}[t!]
	
		\caption{Sequential steps of the $ \textit{\ibfp{}}(h,v) $ process (PRE-PREPARE and PREPARE).}
		\label{algo:ibft-block-finalisation-protocol-pre-prepare-and-prepare}
	
	\UponRules{
		\Gc{$fpgc_1(h,v) \equiv \newline
			\big(\mbox{\hspace{0.4em}}{acceptedBlock}_{h,v} = \bot \:\land\newline
			 \mbox{\hspace{0.9em}}\big( \mbox{\ensuremath{\exists\tsmessage{PRE-PREPARE}{h, r_{h,v}, B}{sv } \in {RM}_v :  sv = \hproposer[] \big) }} \newline\big)
			$
			\label{ln:upon-pre-prepare}
			}{
				\Indp\Indp
				\eIf{$ ({lockedBlock}_{h,v} = \bot \lor {lockedBlock}_{h,v} = B) \land \Valid{\ensuremath{B,{chain_v}[h-1]}} $}
				{
					${acceptedBlock}_{h,v} \gets B$\;
					\Multicastv{\tsmessage{PREPARE}{h, r_{h,v}, \kec{B}}{v}} \To \avhmacro[]\;\label{ln:broadcast-prepare}	
					\label{ln:upon-pre-prepare:end-accept}
				}
				{
					\label{ln:pre-prepare-not-matching-start}
					$moveToNewRoundAndSendRoundChange(r_{h,v}+1, h, v)$\;
					\label{ln:pre-prepare-not-matching-end}
				}
		}

		\Gc{\label{ln:upon-lock}%
			$ fpgc_2(h,v) \equiv \newline \big(
			\mbox{\hspace{0.4em}}{commitSent}_{h,v} = {false} \: \land$\newline 
			\mbox{\hspace{1em}}${{acceptedBlock}_{h,v}}  \neq \bot \: \land$\newline
			\mbox{\hspace{1em}}\SizeOf{$ \{ 
				    \tsmessage{PREPARE}{h, r_{h,v}, \kec{{acceptedBlock}_{h,v}}}{sv} \in {RM}_v  \mid \newline
				    \mbox{\hspace{28em}}sv \in \avhmacro[]\}$ \newline				    
					\mbox{\hspace{4.3em}}$  $} $ \geq \quorum{\nhmacro[]}  \newline \big)$
		}{
			\Indp\Indp${lockedBlock}_{h,v} \gets {acceFptedBlock}_{h,v}$\;\label{ln:set-lock}
		}		
		\Gc{$fpgc_3(h,v) \equiv \newline
			\bigg(
			\mbox{\hspace{0.3em}}{commitSent}_{h,v} = {false} \:\land $ \newline
			\mbox{\hspace{1.0em}}$ {lockedBlock}_{h,v} \neq \bot \: \land $ \newline
			\mbox{\hspace{0.8em}}$ \ensuremath{\Big(\mbox{\hspace{0.0em}}\big( \exists \tsmessage{PRE-PREPARE}{h, r_{h,v}, {lockedBlock}_{h,v} }{sv} \in {RM}_v  : \mbox{\ensuremath{sv = \hproposer[] \big) \:\lor}}}\newline			
			\mbox{\hspace{0.8em}}\mbox{\hspace{0.6em}}\ensuremath{\big(\exists\tsmessage{PREPARE}{h, r_{h,v}, \kec{{lockedBlock}_{h,v} }}{sv}  \in {RM}_v  : sv \in \avhmacro[]  \big)\Big) } \newline \bigg)$
			\label{ln:upon-send-commit}
			}{
				\Indp\Indp
				\Multicastv{}\newline
				{\tsmessage{COMMIT}{h, r_{h,v}, \kec{{lockedBlock}_{h,v}},\cs{{lockedBlock}_{h,v},v}}{v}} \To \avhmacro[]\;
				${commitSent}_{h,v} \gets {true}$\label{ln:set-commit-sent}\;
		}
	}
\end{algorithm}

\begin{algorithm}[t!]
		\caption{Sequential steps of the $ \textit{\ibfp{}}(h,v) $ process (COMMIT and ROUND-CHANGE).}
		\label{algo:ibft-block-finalisation-protocol-commit-and-round-change}
	
	\UponRules{
		\Gc{$ fpgc_4(h,v)  \equiv \newline
		\bigg(\hspace{0.2em}{acceptedBlock}_{h,v} \neq \bot \land \newline
		\mbox{\hspace{1em}}\big(\exists \mathit{CM} \in \mathcal{P}({RM}_v) : $\newline
		\mbox{\hspace{2.5em}}\mbox{$ \SizeOf{\ensuremath{CM}} \geq \quorum{\nhmacro[]} \:\land$}\newline	
		\mbox{\hspace{2.5em}}\mbox{$ (\forall cm \in CM: cm = \tsmessage{COMMIT}{h, r_{h,v}, \kec{{acceptedBlock}_{h,v}},*}{*}) \: \land $} \newline
		\mbox{\hspace{2.5em}}$ (\forall \smessage{cm}{sv},\smessage{cm'}{sv'} \in CM:  cm \neq cm' \implies sv \neq sv') \newline
		\mbox{\hspace{1.2em}}\big) \:\land $
		\newline 
		\mbox{\hspace{1.0em}}$ {finalisedBlockSent}_{h,v} = {false} $\newline$ \bigg) $
		\label{ln:old-send-finalisation-proof-upon-condition}}
		{
			\Indp\Indp
				${finalisedBlockSent}_{h,v} \gets {true}$\label{ln:set-finalised-block-sent}\;
				\eIf{$  \big( \forall \tsmessage{*}{*, *, *,cs}{*} \in CM:  cs\text{ is of the correct size} \big) \land$ \newline
						\mbox{\hspace{-0.6em}}\WellFormedToAddFinalProof{$ {acceptedBlock}_{h,v} $}\label{ln:wellFormedBlock} 
				}
				{
					${lockedBlock}_{h,v} \gets {acceptedBlock}_{h,v}$\label{ln:lock-on-quorum-commits}\;
					\Let ${finalisationProof} \equiv  \{ cs \mid $
					\mbox{$\tsmessage{*}{*, *, *,cs}{sv} \in CM \} $} \label{ln:create-finalisation-proof}\;
					${\fb} \gets \langle {acceptedBlock}_{h,v}, {finalisationProof}\rangle $\;
					\Multicast{$ \tnsmessage{FINALISED-BLOCK}{\fb} $} \To all nodes\label{ln:send-finalise-end}\;					
				}
				{
					\label{ln:fail-finalisation-proof-start}
					$moveToNewRoundAndSendRoundChange(r_{h,v}+1, h, v)$\;
					${lockedBlock}_{h,v} \gets \bot$\label{ln:unlocks}\;
					\label{ln:fail-finalisation-proof-end}				
				}	
		}

		\Gc{
			$fpgc_5(h,v) \equiv \newline
			\big( \exists r_{rc}: $
			\mbox{\hspace{2em}}\mbox{\SizeOf{$ \{ \tsmessage{ROUND-CHANGE}{h, r_{rc}}{sv} \in {RM}_v \mid sv \in \avhmacro[] \} $}
				$\geq f(\nhmacro[])+1 \; \land$} \newline
			\mbox{\hspace{2em}}$r_{rc} > r_{h,v}$\newline $ \big) $
		\label{ln:reception-of-f-plus-1-rc-messages}} {
			\Indp\Indp
			$moveToNewRoundAndSendRoundChange(r_{rc}, h, v)$\;
		}
		\Gc{$ fpgc_6(h,v) \equiv \newline $
			$ \big( now(v) \geq {roundTimerExpiration}_{h,v}[r_{h,v}] \big)$\label{ln:expiry}} {
			\Indp\Indp$moveToNewRoundAndSendRoundChange(r_{h,v}+1, h, v)$\;
		}
		\Gc{$ fpgc_7(h,v) \equiv \newline
			\big(\exists r_{rc}: $\newline 
			\mbox{\hspace{1em}}\SizeOf{$ \{ \tsmessage{ROUND-CHANGE}{h, r_{rc}}{sv} \in {RM}_v \mid sv \in \avhmacro[] $\newline
				\mbox{\hspace{4.2em}}}
			$\geq \quorumoptimal({\nhmacro[]}) \:\land$ \newline  
			\mbox{\hspace{1em}}$\big( r_{rc} > {r}_{h,v} \lor 
			(r_{rc} = {r}_{h,v} \land 
			{roundAlreadyStarted}_{h,v} = {false})\big)$\newline $ \big) $
			\label{ln:receive-quorum-round-change-messages}}{
			\Indp\Indp$StartNewRound(r_{rc},h,v)$\;
		}	
	}
\end{algorithm}

This sub-section describes the execution of a generic $h$ instance of the \ibfp{} for validator $v$ which, described by \Cref{eqn:ibfp-h-v,eqn:ibftExecution-v}, corresponds to the execution of the $ fpInit(h,v) $ step followed by the recursive execution of the non-deterministic choice between the guarded commands $ \big( fpgc_i(h,v)\big)_{i\in [7]} $.

The \ibfp{} is organised in rounds, starting from round 0, where validators progress to the next round once they suspect that in the current round they will not be able to decide on the block to be included in the finalised block with height $h$.
Both in the pseudocode and here, the current round for the $h$-th instance of the \ibfp{} for validator $v$ is denoted by $r_{h,v}$.

For each round, one of the validators is selected to play the role of block proposer. 
This selection is operated by the evaluation of \proposermacro[] where \Proposer{$ \cdot,\cdot $} is  a deterministic function of the chain of blocks from the genesis block until the block with height $h-1$ and the current round number.

The pseudocode at \cref{ln:avibfp,ln:hproposer,ln:nhmacro} introduces the following macros:
\begin{itemize}
	\item \nhmacro[]: number of validators for the $h$-th instance of the \ibfp{} for validator $v$;
	\item \avhmacro[]: validators for the $h$-th instance of the \ibfp{} for validator $v$;
	\item \hproposer[]: proposer for round $r_{h,v}$ of the $h$-th instance of the \ibfp{} for validator $v$.
\end{itemize}
These macros are used both in the pseudocode and in this section to simplify the notation when describing the $h$-th instance of the \ibfp{} for validator $v$. 

For the purpose of this work, we do not define the proposer selection function, but we state that it ensures that all of the validators for the $h$-th instance of the \ibfp{} are selected for any sequence of $\nhmacro[]$ consecutive rounds.

During the \ibfp{}, when specific conditions are met, a validator $v$ can \lock[] on a block $B$, which corresponds to the pseudocode in \crefrange{algo:ibft-block-finalisation-protocol-start-and-init}{algo:ibft-block-finalisation-protocol-commit-and-round-change} setting ${lockedBlock}_{h,v}$ to a value different from $\bot$ (see \cref{ln:set-lock}). 
However, regardless of whether a validator $v$ is \lock[ed] or not, $v$ always \unlock[s] (sets ${lockedBlock}_{h,v}$ to $\bot$) when a new \ibfp{} instance starts (see \cref{ln:unlock-at-initialisation}).

\begin{sloppypar}
As specified by the $ {StartNewRound} $ procedure (\cref{ln:start-new-round}), at the beginning of the current round $r_{h,v}$, if $v$ is the selected block proposer for round $r_{h,v}$, then $v$ multicasts a Pre-Prepare message \tsmessage{PRE-PREPARE}{h, r_{h,v}, B}{p_{r_{h,v}}} to all validators (including itself) where, if $v$ is \emph{locked} on a block then $B={lockedBlock}_{h,v}$, otherwise $B$ can be any valid block for height $h$. 
The pseudocode uses \CreateNewProposedBlock{h, v} to represent the creation of a new block with height $h$ by validator $v$. 
Honest validators employ a fair transaction selection algorithm to decide which transactions to include in the next block.
The definition of such algorithm is outside the scope of this work.
\end{sloppypar}

As specified by \crefrange{ln:upon-pre-prepare}{ln:upon-pre-prepare:end-accept}, a validator $v$ accepts a Pre-Prepare message \tnsmessage{PRE-PREPARE}{h_{pp}, r_{pp}, B} if and only if all of the following conditions are met:
\begin{itemize}
	\item $v$ is currently running the \ibfp{} instance   $ h_{pp} $, i.e $h_{pp} = h$ ;
	\item $v$ is in round $r_{pp}$, i.e. $r_{pp} = r_{h,v} $;
	\item the message is signed by the selected proposer for round $r_{h,v}$ and instance $h$ of the \ibfp{};
	\item $v$ has not already accepted a Pre-Prepare message for round $ r_{h,v} $ in the $h$-th instance of the \ibfp{}, i.e. $ {acceptedBlock}_{h,v} = \bot $;
	\item $v$ is not \lock[ed] on a block different from $B$;
	\item block $B$ is a valid block for height $h$.
\end{itemize}
\begin{sloppypar}
	When a validator $v$ accepts a Pre-Prepare message, it multicasts a Prepare message \tsmessage{PREPARE}{h, r_{h,v}, \kec{B}}{v} (see \cref{ln:broadcast-prepare}) to all validators.
	We refer to the block included in a Pre-Prepare message for round $ r $ that validator $ v $ has accepted in instance $ h $ as the \emph{accepted block for round $ r $}.
	Also, if validator $ v $ accepts message \tnsmessage{PRE-PREPARE}{h, r, B}, by extension we say that validator $ v $ \emph{accepts} block $ B $ in round $ r $.
	Both in the pseudocode and in the remainder of this section, $ {acceptedBlock}_{h,v} $ corresponds to the block accepted by validator $ v $ for the current round of the instance $ h $ of the \ibfp{}.
	$ {acceptedBlock}_{h,v} $ is equal to $ \bot $ if validator $ v $ has not accepted any Pre-Prepare message for the current round.
\end{sloppypar}

\begin{sloppypar}
	As specified by the guarded command at \cref{ln:upon-lock}, if (i) validator $ v $ has accepted a block in the current round ($r_{h,v}$) of instance $ h  $ of the \ibfp{}, i.e. $  {acceptedBlock}_{h,v} \neq \bot  $ and (ii) $ v $ has received at least $\quorum{n_{h,v}} $ Prepare messages for the current round of instance $ h $ of the \ibfp{} with Keccak hash matching the Keccak hash of the accepted block for the current round, then validator $v$  \lock[s] on the accepted block for the current round, i.e $ lock_{h,v} \gets  {acceptedBlock}_{h,v} $.
\end{sloppypar}

As specified by the guarded command at \cref{ln:upon-send-commit}, if validator $ v $ has \lock[ed] on a block, i.e. $ lock_{h,v} \neq \bot $ then the first time that one of the conditions listed below is verified, $v$ multicasts a Commit message \tsmessage{COMMIT}{h,r_{h,v},\kec{lock_{h,v}},\cs{lock_{h,v},v}}{v} to all validators (including itself), where \cs{lock_{h,v},v}, called \emph{commit seal}, corresponds to the signature of $v$ over the \lock[ed] block $lock_{h,v}$.
\begin{itemize}
	\item $v$ has received a Pre-Prepare message for the locked block $lock_{h,v}$ and current round of the $ h $ instance of the \ibfp{};
	\item $v$ has received a Prepare message for a Keccak hash matching the Keccak hash of the \lock[ed] block $\kec{lock_{h,v}}$ and current round of the $ h $ instance of the \ibfp{}.
\end{itemize}
The pseudocode uses the state variable ${commitSent}_{h,v}$ to indicate that, for a given round, the Commit message is sent at most once. Indeed, ${commitSent}_{h,v}$ is set to $true$ at \cref{ln:set-commit-sent} and reset to $false$ in the ${StartNewRound}$ procedure at \cref{ln:reset-commit-sent}.
In the IBFT implementation~\cite{jpmorgangithub}, \cs{B, v} is actually a signature over a modified version of block $B$, but in this work we consider the simplified definition provided above as the differences between this definition and the complete one do not affect the results presented here.

The IBFT protocol also includes the following optimisation that, for brevity, we have omitted from the pseudocode description. 
Commit messages are treated as Prepare messages when evaluating the conditions for \lock[ing] and sending Commit messages.
For example, if $\quorum{n_{h,v}} = 4$ and $v$ accepts a Pre-Prepare message for height $h$, round $r_{h,v}$ and block $B$, 2 Prepare messages for height $h$, round $r_{h,v}$ and block with Keccak hash matching $\kec{B}$, and 2 Commit messages for the same height, round and block hash, then $v$ \lock[s] on $B$ and sends a Commit message for height $h$, round $r_{h,v}$ and block hash $\kec{B}$.

\begin{sloppypar}
	Additionally, as specified by the guarded command at \cref{ln:send-commit-for-old-blocks} in \cref{algo:ibft-protocol}, if a validator $ v $ receives a \tnsmessage{PRE-PREPARE}{h_{pp}, r_{pp}, B_{pp}} message and all of the following conditions are verified, then $ v  $  multicasts a Commit message \tsmessage{COMMIT}{h_{pp},r_{pp},\kec{B_{pp}},\cs{B_{pp},v}}{v} to all validators (including itself):
\end{sloppypar}
\begin{itemize}
	\item the Pre-Prepare message is for a previous block height, i.e. $ h_{pp} < h_v $;
	\item the Pre-Prepare message is signed by the selected proposer for round $ r_{pp} $ and instance $ h_{pp} $;
	\item $ B_{pp} $ matches the Ethereum block contained in the finalised block with height $ h_{pp} $ included in the local blockchain of validator $ v $.
\end{itemize}
As specified by \crefrange{ln:send-finalise-start}{ln:send-finalise-end}, the first time that all of the following conditions are verified for validator $v$, (i) $v$ \lock[s] on the accepted block ${acceptedBlock}_{h,v}$ (\cref{ln:lock-on-quorum-commits}) and (ii) broadcasts a finalised block including the accepted block ${acceptedBlock}_{h,v}$ and related finalisation proof:
\begin{itemize}
	\item $ v $ has accepted a block in the current round ($r_{h,v}$) of instance $ h  $ of the \ibfp{}, i.e. $  {acceptedBlock}_{h,v} \neq \bot $;
	\item $v$ has received Commit messages,  from at least $\quorum{n}$ distinct \avhmacro[], for height $h$, current round $r_{h,v}$ of instance $ h $ and Keccak hash matching the Keccak hash of the accepted block \kec{{acceptedBlock}_{h,v} };
	\item all of the commit seals received (as part of Commit messages) are of the correct size;
	\item the accepted block $ {acceptedBlock}_{h,v}  $ is well formed to allow adding the finalisation proof to it.
\end{itemize} 
As indicated by \cref{ln:create-finalisation-proof}, the finalisation proof includes all of the commit seals included in Commit messages received for the current round ($r_{h,v}$) of instance $ h $ and Keccak hash matching \kec{{acceptedBlock}_{h,v} }.
The pseudocode uses the state variable ${finalisedBlockSent}_{h,v} $ to trigger the transmission of a finalised block only the first time that the conditions listed above are met.
${finalisedBlockSent}_{h,v} $ is set at \cref{ln:set-finalised-block-sent}  and reset in the ${StartNewRound}$ procedure at \cref{ln:reset-finalised-block-sent}.

In alignment with PBFT, IBFT relies on a round change sub-protocol to detect whether the selected proposer may be Byzantine and causing the protocol to never terminate. 
When one of the conditions listed below is satisfied for validator $v$ while in round $r_{h,v}$, $v$ moves to a new round $r'$ and multicasts a Round-Change message \tsmessage{ROUND-CHANGE}{h, r'}{v} to all validators (including itself).
\begin{itemize}
	\item \textbf{Round Timer Expiry} (\cref{ln:expiry}). Expiration of the round timer started by each validator at the beginning of every round (see \cref{ln:start-timer}). 
	The length of the round time is exponential to the round number. 
	In this case $v$ moves to round $r' = r_{h,v}+1$.
	\item \textbf{Pre-Prepare message not matching \lock[ed] block} (\crefrange{ln:pre-prepare-not-matching-start}{ln:pre-prepare-not-matching-end}). Reception of a Pre-Prepare message sent by the selected proposer for round $r_{h,v}$ with proposer block not matching the block on which validator $v$ is \lock[ed].
	In this case $v$ moves to round $r'=r_{h,v}+1$.
	\item \textbf{Reception of $f(n)+1$ Round-Change messages for future round} (\cref{ln:reception-of-f-plus-1-rc-messages}).  Reception of $f(n) + 1$ Round-Change messages for instance $h$ and round $r'$ with $r' > r_{h,v}$.
	In this case $v$ moves to round $r'$.
	\item  \textbf{Failure in creating the finalisation proof} (\crefrange{ln:fail-finalisation-proof-start}{ln:fail-finalisation-proof-end}). 
	$v$ has received at least \quorum{n_{h,v}} Commit messages for instance $h$, current round and Keccak hash \kec{{acceptedBlock}_{h,v} }  and at least one of the following conditions is verified:
	\begin{itemize}
		\item at least one of the commit seals included in the Commit messages received by a validator for round $r_{h,v}$ and instance $h$ is of the wrong size;
		\item the accepted block ${acceptedBlock}_{h,v} $ is not formatted correctly and does not allow adding the finalisation proof to it.
	\end{itemize}
	If any of these conditions are verified, then $v$ also \unlock[s] before moving to the next round (see \cref{ln:unlocks}).
\end{itemize}
When validator $v$ moves to a new round $r'$, the state variable ${acceptedBlock}_{h,v}$ is reset to $\bot$ (see \cref{algo2:set-accept-pre-prepare-1}).
As it can be noted from the pseudocode, moving to a new round (\cref{ln:move-to-new-round}) does not imply starting a new round (\cref{ln:start-new-round}).
As specified by guarded command $ fpgc_7(h,v) $ at \cref{ln:receive-quorum-round-change-messages}, when a validator $v$ receives either \quorum{n_{h,v}} Round-Change messages for instance $h$  and round $ r' $ with $r' > r_{h,v}$ or \quorum{n_{h,v}} Round-Change messages for instance $h$  and round $r'$ matching the current round (i.e. $r' = r_{h,v}$) but the current round has yet to be started, then $v$ starts round $r'$.\\
Starting round $r'$ includes executing the following actions:
\begin{itemize}
	\item $v$ starts the round timer for round $r'$ with length $\RoundTimerTimeout{\ensuremath{r'}}$;
	\item if $v$ is the proposer for round $r'$, then $v$ multicasts the following  Pre-Prepare message to all validators: \tsmessage{PRE-PREPARE}{h, r', B'}{p_{r'}} where, if $p_{r'}$ is \emph{locked} on block $B$ then $B'=B$, otherwise $B'$ can be any valid block for instance $h$.
\end{itemize}

From here on the protocol proceeds as described above.

\section{Definitions}\label{sec:definitions}
In this section we provide a series of definitions that will be used in the presentation of our persistence and liveness analysis in \crefrange{sec:safety-analysis}{sec:liveness-analysis}.

We define $t$-Byzantine-fault-tolerant persistence as follows where $ t $ can be either a constant value or a deterministic function of some sort.
\begin{definition}[$t$-Byzantine-fault-tolerant persistence]
	The IBFT protocol ensures $t$-Byzantine-fault-tolerant persistence if and only if the following statement is true:
	provided that no more than $t$ validators are Byzantine, the IBFT protocol guarantees the persistence property of distributed permissioned transactions ledgers (see \cref{def:robustnes}). 
\end{definition}

The aim of the $h$-th instance of the \ibfp{} is to have all honest validators to eventually decide on the block to be included in the finalised block with height $h$ and broadcast the  finalised block to all the nodes.\\
In the context of the \ibfp{} we define safety as follows:
\begin{definition}[$t$-Byzantine-fault-tolerant safety for the \ibfp{}]\label{def:t-tolerate-ibfp}
	The \ibfp{} ensures $t$-Byzantine-fault-tolerant safety if and only if it guarantees the validity of the following statement:
	in the presence of no more than $t$ Byzantine validators and provided that all honest validators agree on the same validator set, the protocol ensures that instance $ h $ of the \ibfp{} can only produce a valid finalised block for height $ h $ and that no two valid finalised blocks including different blocks can ever be produced by any instance $ h $ of the \ibfp{}.
\end{definition}

Our following definitions of optimal Byzantine-fault-tolerant safety follows directly from the known upper-limit on the number of Byzantine nodes that a consensus protocol operating in an eventually synchronous network can withstand \cite{Dwork:1988:CPP:42282.42283}:
\begin{definition}[Optimal Byzantine-fault-tolerant safety threshold for the \ibfp{}]\label{def:optimal-byz-fault-tol-safety-for-ibfp}
	The \ibfp{} guarantees optimal Byzantine-fault-tolerant safety threshold provided that for any instance $h$ it guarantees $ f(n_h) $	-Byzantine-fault-tolerant safety where $n_h$ is the number of validators for the $h$-th instance of \ibfp{}.
\end{definition}

The following two definitions are related to the liveness property of the IBFT protocol.
\begin{definition}[$t$-Byzantine-fault-tolerant liveness]
	The IBFT protocol ensures $t$-Byzantine-fault-tolerant liveness if and only if the following statement is true:
	provided that no more than $t$ validators are Byzantine, the IBFT protocol guarantees the liveness property of distributed permissioned transactions ledgers (see \cref{def:robustnes}). 
\end{definition}

\begin{definition}[$t$-Byzantine-fault-tolerant weak-liveness of the \ibfp{}]
	The \ibfp{} guarantees $t$-Byzantine-fault-tolerant weak-liveness if and only if, provided that no more than $t$ validators are Byzantine, it guarantees that for any $h$ instance of the \ibfp{} at least one honest validator will eventually be able to produce a valid finalised block for height $h$.
\end{definition}

\section{Persistence Analysis}\label{sec:safety-analysis}
In this section we analyse the persistence property of the IBFT protocol that in conjunction with the liveness property determines that overall robustness of the IBFT protocol as defined in \cref{def:robustnes}.

Any proof omitted in this section can be found in \ref{ap:proofs-for-persistence analysis}.

\begin{lemma}\label{lem:ibft-protocol-safety-reverse}
	If the \ibfp{} does not guarantee  $t$-Byzantine-fault-tolerant safety, then the IBFT protocol does not guarantee  $t$-Byzantine-fault-tolerant persistence. 
\end{lemma}
\begin{proof}
	Assume that $v$ and $v'$ are two honest nodes such that the height of their local blockchain is $h$.
	According to the Lemma, in presence of $t$ Byzantine validators for height $h$, the \ibfp{} can produce two different valid finalised blocks including blocks $B$ and $B'$ respectively,  with $B \neq B'$.
	Assume that transactions $T$ and $T'$, with $T \neq T'$, are included at position $i$ of blocks $B$ and $B'$ respectively.
	This is possible as $B \neq B'$.
	If this happens before GST (see the Network Model definition in \cref{sec:system-model}), then $v$ could receive the finalised block including $B$ while $v'$ could receive the finalised block including $B'$.
	Since both finalised blocks are valid, $v$ adds the finalised block including $B$ in position $h$ of its local blockchain while $v'$ adds the finalisation block including $B'$ in position $h$ of its local blockchain. This equates to $T$ being added in position $(h,i)$ of the local ledger of $v$ and $T'$ being added in the same position $(h,i)$ of the local ledger of $v'$, with $T \neq T'$.
\end{proof}

\begin{lemma}\label{lem:transaction-t-uniquiness}
	A transaction $ T $ cannot appear in two different positions of the local blockchain of an honest node. 
\end{lemma}

\begin{lemma}\label{lem:if-an-honest-node-sets-hv-to-h-then-it-has-all-previous-blocks}
	One of the necessary preconditions for an honest validator $ v $ to set $ h_v $ to $ h $ is that all blocks with height $ < h $ have been added to its local blockchain.
\end{lemma}

\begin{proof}
	There is no line of pseudocode that removes a finalised block from the local blockchain.
	The only point where finalised blocks are added to the local blockchain is \cref{ln:add-block} which adds a finalised block in position $ h_v $ of the local chain.
	$ h_v $ is set to $ 0 $ at the initialisation ($ fpInit(v) $) and the only point where $ h_v $ is changed is \cref{ln:increase-h-v} where it is incremented by one.
	The observation that \Cref{ln:increase-h-v} is executed any time that \Cref{ln:add-block} is executed concludes the proof.
\end{proof}

\begin{lemma}\label{lem:if-an-honest-nodes-start-instance-h-then-it-has-all-previous-blocks}
	One of the necessary preconditions for an honest validator $ v $ to start instance $ h $ of the \ibfp{} is that all blocks with height $ < h $ have been added to its local blockchain.
\end{lemma}

\begin{proof}
	According to \cref{ln:start-ibfp-1,ln:start-ibfp-2} of \cref{algo:ibft-protocol}, in an honest node starts instance $ h $ of the \ibfp{}, then $ h_v = h $.
	This and \cref{lem:if-an-honest-node-sets-hv-to-h-then-it-has-all-previous-blocks} concluded the proof.
\end{proof}

%
%

%
%
%
\begin{lemma}\label{lem:if-ibfp-is-t-BFT-then-no-two-valid-blocks-for-the-same-height-are-ever-created}
	If the \ibfp{} guarantees $ t $-Byzantine-fault-tolerance safety and for any possible sequence of finalised blocks $ \fb_0,\ldots,\fb_{h\geq 0} $ stored in the local blockchain of any honest validator, the number of Byzantine validators in \AV{$ \fb_0,\ldots,\fb_h $} is no more than $ t $, then no two valid finalised blocks for the same height but including two different blocks are ever produced by any node, whether Byzantine or honest.
\end{lemma}
\begin{proof}
	The proof is by induction on the length $ h $ of the sequence of finalised blocks.
	Before proceeding with the proof, we shall remind that (a) each validator computes the set of validators for their own instance $ h $ of the \ibfp{} as a deterministic function of the chain of finalised blocks from the genesis block to the block with height $ h-1 $.
	
	\textbf{Base Case $ h = 1. $} 
	Since we assume that all honest nodes have the same genesis block, all honest validators will agree on same set of validators for first instance of the \ibfp{}. The assumptions that no more than $ t $ of these validators are Byzantine and that the \ibfp{} guarantees $ t $-Byzantine-fault-tolerant safety imply the Lemma for the base case.
	
	\textbf{Inductive Case.}
	For the inductive case we assume that the Lemma is true for all $ h' $ such that $ 1 \leq h' < h $ and then we prove that the Lemma also holds for $ h $.
	By \cref{lem:if-an-honest-nodes-start-instance-h-then-it-has-all-previous-blocks}, the assumption for the inductive case and property (a), all honest validators that start instance $ h $ of the \ibfp{} agree on the same set of validators.
	The assumption that the \ibfp{} guarantees $ t $-Byzantine-fault-tolerant safety imply the Lemma for the inductive case.
\end{proof}

\begin{lemma}\label{lem:common-prefix}
	If for any possible sequence of blocks $ \fb_0,\ldots,\fb_{h\geq0} $ stored in the local blockchain of any honest validator, the number of Byzantine validators in \AV{$ \fb_0,\ldots,\fb_h $} is no more than $ t $, then the local blockchains of any two honest nodes $ v $ and $ v' $ agree on any block with height $ h  $ such that  $0 \leq h < \min(h_v, h_{v'}) $. 
\end{lemma}

\begin{proof}
	For a given $ h $, such that $0 \leq h < \min(h_v, h_{v'}) $, \cref{lem:if-an-honest-node-sets-hv-to-h-then-it-has-all-previous-blocks} guarantees that both $ v $ and $ v' $ have added a finalised block with height $ h $ to their respective local blockchains, and \cref{lem:if-ibfp-is-t-BFT-then-no-two-valid-blocks-for-the-same-height-are-ever-created} implies that both finalised blocks include the same block.
	Since  all honest nodes have the same genesis block and honest nodes only add valid finalised block to their local blockchains (see \cref{ln:finalise-block-start}), the Lemma is proven.
\end{proof}


%
%

	\begin{lemma}\label{cor:guarantee-of-persitence-i}
	The IBFT protocol guarantees condition (i) of $t$-Byzantine-fault-tolerant persistence if and only if the \ibfp{} guarantees $t$-Byzantine-fault-tolerant safety. 
\end{lemma}
\begin{proof}
	The forward direction is proved by \cref{lem:transaction-t-uniquiness,lem:common-prefix}.
	The reverse direction is proved by \Cref{lem:ibft-protocol-safety-reverse}.
\end{proof}

\Cref{def:optimal-byz-fault-tol-safety-for-ibfp} and \cref{cor:guarantee-of-persitence-i} lead to the following definition of optimal Byzantine-fault-tolerant persistence threshold.
\begin{definition}[Optimal Byzantine-fault-tolerant persistence threshold for the IBFT protocol]\label{def:optimal-Byzantine-fault-tolerance-persistence}
	The IBFT protocol guarantees optimal Byzantine-fault-tolerant persistence threshold provided that all possible instances of the \ibfp{}  guarantee optimal Byzantine-fault-tolerant safety threshold.
\end{definition}

\begin{lemma}\label{lem:n-1-geq-quorum}
	$n-1 \geq \quorum{n}$ for any $n \geq 4$.
\end{lemma}

\begin{lemma}\label{lem:ibfp-not-safe}
	The \ibfp{} does not guarantee Byzantine-fault-tolerant safety when operating in an eventually synchronous network (see Network Model definition in \cref{sec:system-model}).
\end{lemma}
\begin{proof}	
	The proof is by counterexample.
	For this proof we use the inductive assumption of \cref{lem:common-prefix} which states that the local blockchains of all honest nodes are identical until the finalised block with height $h-1$.
	Therefore, since the set of validators for the $h$-th instance of the \ibfp{} is a function of the local blockchain until the block with height $h-1$, this set is identical amongst all honest validators.
	We denote the total number of validators for the $h$-th instance of the \ibfp{} with $n_h$.
	
	Below we provide a possible sequence of events that may occur before GST which leads two honest nodes to create two valid finalised blocks for the same height $h$ that include two different blocks.
	We assume $f(n_{h}) \geq 1$ which implies $n_h \geq 4$.
	Let $v$ be an honest validator, $W$ be a set of size $n_{h}-1$ including all validators except $v$ and $W_{honest}$ be the subset of all honest validators included in $W$.
	\begin{enumerate}
		\item Let $r$ be any round such that no honest validator has \lock[ed] on any block.
		This is always the case for $r=0$.
		The proposer for round $r$, $p_r$, multicasts a \tsmessage{PRE-PREPARE}{h, r, B}{p_r} message to all validators (including itself).
		\item All validators (both Byzantine and honest)  receive the Pre-Prepare message sent by the proposer and multicast  the corresponding \tnsmessage{PREPARE}{h, r, \kec{B}} message to all validators (including themselves).
		\item  \label{itm:ibfp-not-safe-step-3}
		\begin{sloppypar}
			All honest validators ($W_{honest} \cup \{v\}$) receive the $n_{h}$ Prepare messages sent by all validators.
			Since $n_h \geq \quorum{n_h}$, all honest validators \lock[] on block $B$  and multicast a \tsmessage{COMMIT}{h, r, \kec{B}, \cs{B,sender}}{sender} message to all other validators.
		\end{sloppypar}
		\item \label{itm:ibfp-not-safe-step-4} Byzantine validators also send a well-formed \tsmessage{COMMIT}{h, r, \kec{B}, \cs{B,sender}}{sender} message to honest validator $v$.
		
		\item \label{itm:ibfp-not-safe-step-5} However, Byzantine validators send Commit messages containing a commit seal of the wrong size to all honest validators included in set $W_{honest}$.
		\item As result of events \labelcref{itm:ibfp-not-safe-step-3,itm:ibfp-not-safe-step-4,itm:ibfp-not-safe-step-5}, each of the $n_h$ validators (whether Byzantine or honest) has sent a Commit message to all the honest validators, $v$ will receive only well-formed Commit messages, whereas the Commit messages sent by Byzantine validators that the validators in $W_{honest}$ will receive will include a commit seal of the wrong size.
		Assume, without loss of generality, that at least one of the Commit messages with the commit seal of the wrong size is included in the  first \quorum{n_{h}} Commit messages received by each of the honest validators included in set $W_{honest}$.
		Therefore, once each validator included in $W_{honest}$ receives \quorum{n_{h}} Commit messages, since at least one of them has a commit seal of the wrong size, it \unlock[s] and sends a \tnsmessage{ROUND-CHANGE}{h, r'=r+1} message to all validators (including itself).\\
		In contrast, $v$ only receives Commit messages with valid commit seals and therefore it creates a valid finalised block including block $B$ and broadcasts it to all the nodes.
		However, we assume that the finalised block created by $v$ will not be received by any validator before the end of the last event of this sequence of events.
		This is possible as GST has yet to be reached.
		\label{item:create-finalised-for-B}\label{item:ulock-event}
		\item Like validators in $W_{honest}$, all Byzantine validators in $W$ also send a \tnsmessage{ROUND-CHANGE}{h, r'} message to all validators.
		\item All validators in $W_{honest}$ receive all of the Round-Change messages for round $r'$ sent by all validators in $W$, both Byzantine and honest.
		According to Lemma~\ref{lem:n-1-geq-quorum}, since $n_h \geq 4$, the following holds $n_h-1 \geq \quorum{n_h}$. 
		Therefore all validators in $W_{honest}$ start round 1.
		\item Assume, without loss of generality, that the proposer of round $r'=r+1$, $p_{r'}$, is not $v$.
		$p_{r'}$ multicasts a \tsmessage{PRE-PREPARE}{h, r', B'}{p_{r'}} message with $B' \neq B$ to all validators.
		This is possible as $p_{r'}$ is either Byzantine or honest but included in set $W_{honest}$ and has therefore \unlock[ed] at event \ref{item:ulock-event}.
		\item All validators in $W$ receive the Pre-Prepare message sent by $p_{r'}$ and therefore broadcast a \tnsmessage{PREPARE}{h, r', \kec{B'}} message to all validators (including themselves).
		\item \begin{sloppypar}
			All honest validators included in $W_{honest}$ receive the $n_{h}-1$ Prepare messages sent by all validators included in set $W$.
			According to Lemma~\ref{lem:n-1-geq-quorum}, since $n_h \geq 4$, the following holds $n_h-1 \geq \quorum{n_h}$. 
			Therefore all honest validators included in set $W_{honest}$ \lock[] on block $B'$  and multicast a \tsmessage{COMMIT}{h, r', \kec{B'}, \cs{B',sender}}{sender} message to all other validators.	
			Byzantine validators also multicasts the same Commit message.
			All Commit messages sent for this round include a valid commit seal.
		\end{sloppypar} 
		
		\item All honest validators included in set $W_{honest}$ receive the $n_{h}-1$ Commit messages sent by all validators included in set $W$.
		According to Lemma~\ref{lem:n-1-geq-quorum}, since $n_h \geq 4$, the following holds $n_h-1 \geq \quorum{n_h}$. 
		Therefore, all honest validators in set $W_{honest}$ create a valid finalised block for block $B'$.\label{item:create-finalised-for-B-prime}
	\end{enumerate}
	This concludes the proof as two valid finalised blocks including different blocks ($B$ and $B'$) have been created at events \ref{item:create-finalised-for-B} and \ref{item:create-finalised-for-B-prime}.
\end{proof}

\begin{lemma}\label{lem:ibft-violates-codition-i-of-persistence}
	The IBFT protocol does not guarantee condition (i) of $t$-Byzantine-fault-tolerant persistence.
\end{lemma}
\begin{proof}
	Direct consequence of \cref{cor:guarantee-of-persitence-i,lem:ibfp-not-safe}.
\end{proof}

%

\begin{lemma}\label{lem:n->2-f(n)}
	$\forall n \geq 1: n > 2\cdot f(n) $
\end{lemma}

\begin{lemma}\label{lem:ibft-violates-codition-ii-of-persistence}
	The IBFT protocol does not guarantee condition (ii) of $t$-Byzantine-fault-tolerant persistence even if the \ibfp{} guarantees $t$-Byzantine-fault-tolerant safety and $t$-Byzantine-fault-tolerant weak-liveness.
\end{lemma}

\begin{proof}
	The proof is by counterexample. 
	Indeed, we assume that the \ibfp{} guarantees both $t$-Byzantine-fault-tolerant safety and $t$-Byzantine-fault-tolerant weak-liveness and show a sequence of events that, if occurring before GST, lead the IBFT protocol to a state where condition (ii) of the persistence property is violated.
	We denote the number of validators for the $h$-th instance of the \ibfp{} with $n_h$.
	Considering that no Byzantine-fault-tolerant safety or weak-liveness can be ensured if $t=0$, we assume $t \geq 1$. 
	Since $f(n_h)$ is the upper limit for $t$, this implies $f(n_h) \geq 1$,  which in turns implies $n_h \geq 4$.
	\Cref{lem:n->2-f(n)} proves that $n_h > 2\cdot f(n_h)$ which implies that $n_h > 2 \cdot t$.
	We also assume that all validators are running the $h$-th instance of the \ibfp{}.
	Let $ V $ be a set of $t$ honest validators and $W$ be the complement of $V$ for instance $h$ of the \ibfp{}.
	$W$ is therefore of size $n_h-t$ and may contain up to $t$ Byzantine validators.
	\begin{enumerate}
		\item Validators in set $W$ produce a finalised block, say \fb{}, but no messages sent by validators in $V$ while in instance $h$, including the Finalised-Block message for \fb{}, will ever be received by validators in set $V$.
		This is possible as (i) GST has yet to be reached and (ii) we assume that the \ibfp{} guarantees $t$-Byzantine-fault-tolerant weak-liveness which implies that a set of $n_h -t$ validators can produce a finalised block without needing to communicate with any validators outside that set.
		Consequently, validators in set $W$ terminate their $h$-th instance of the \ibfp{} and start the $h+1$-th instance, while all validators in set $V$ continue to run the $h$-th instance of the \ibfp{}.
		Without loss of generality, assume that finalised block \fb{} does not change the validator set, i.e. the validator set for instance $h+1$ matches the validator set of instance $h$ which implies that $n_{h+1} = n_h$.
		\item $t$ of the validators in set $W$ stop communicating and never restart.
		This is possible as we assume that up to $t$ of the validators in set $W$ may be Byzantine and therefore act arbitrarily.
		As consequence of this, out of the $n_h$ validators in the validator set for instance $h+1$, $2 \cdot t$ of them do not participate in this instance as $t$ validators are still running instance $h$ and $t$ Byzantine validators have stopped and will never restart.
		Since $2\cdot t>t$, weak-liveness of the \ibfp{} for instance $h+1$ does not hold and therefore there is no guarantee that any new finalised block in instance $h+1$ is produced under this condition.
		
		Also, the remaining $n_h -2\cdot t$ honest validators of set $W$, running the $h+1$-th instance of the \ibfp{} send Commit messages for instance $h$ only if they receive a Pre-Prepare message for instance $h$ with a proposed block matching the block included in \fb{}.
		Since none of the messages transmitted by validators in set $W$ while they were running instance $h$ has been received by validators in set $V$, there is no guarantee that a validator in set $V$ will eventually send a Pre-Prepare message with a proposed block matching the block included in \fb{} and therefore there is no guarantee that any of the honest validators in set $W$ will ever again participate in instance $h$ by sending Commit messages.
		
		Since $n_h-t$ of the $n_h$ validators do not actively participate in instance $h$,  and $n_h>2\cdot t \rightarrow n_h -t>t$ implies that weak-liveness of the \ibfp{} does not hold for instance $h$ either, there is no guarantee that any finalised block in instance $h$  will ever be produced by validators in set $V$.
	\end{enumerate}
	In the system state resulting from the last event of the sequence of events presented above there is no guarantee that any new finalised block for height $h$ will ever be broadcast. Considering that we assume that finalised block \fb{} produced by validators in set $W$ will never be received by validators in set $V$, there is no guarantee that any of the transactions included in the block included in the finalised block \fb{} will ever be added to the transaction ledger of any validator in set $V$.
\end{proof}


\begin{theorem}\label{lem:ibft-not-safe}
	The IBFT finalisation protocol is not Byzantine-fault-tolerant when operating in an eventually synchronous network.
\end{theorem}

\begin{proof}
	Direct consequence of \cref{lem:ibft-violates-codition-i-of-persistence,lem:ibft-violates-codition-ii-of-persistence}.
\end{proof}

\section{Modification \ibftmonep{}: Achieve optimal Byzantine-fault-tolerance persistence}\label{sec:modification-ibftmonep-achieve-optimal-byzantine-fault-tolerance-persistence}

\newcommand{\changed}[1]{\colorbox{yellow}{#1}}
\newcommand{\deleted}[1]{\changed{\sout{#1}}}
\begin{algorithm}[t!]
	\setcounter{AlgoLine}{0}
	\Functions{ 
		\colorbox{yellow}{\lFn{\mbox{\hspace{-0.4em}}\hl{$ \quorumoptimal(n) $}}{\hl{$ \left \lceil \frac{2n}{3}\right \rceil$}}}\;
		\Fn{${isValidFinalisationProof}(\fb,v)$}{
			\SizeOf{$ \{
				cs \in \ExtractFinalisationProof{\fb} \mid \newline$
				\mbox{\hspace{3em}}$\EthAddressRecover{\ensuremath{\kec{\ExtractBlock{\fb}},cs }}\in \avibfp[v]\}$\newline
				\mbox{\hspace{1.3em}}} $\geq \colorbox{yellow}{\ensuremath{\quorumoptimal(\nmacro[v])}}$
		}	
		\Fn{${isValidFinalisedBlock}(\fb, v)$}{
			${isValidFinalisationProof}(\fb,v) \:\land$\newline
			\mbox{\hspace{-2em}}\Valid{${\ExtractBlock{\fb},}$
				$\ExtractBlock{\ensuremath{{chain}_v[h_v-1]}}$}
			
	}}
	\Init{
		\Fn{$ ibftInit(v)  $}
		{
			${chain}_v[0] \gets \textit{genesis block}$\;
			$h_v \gets 1 $\;
			\deleted{$ oldPrePrepareMessagesAlreadyProcessed_v \gets \{\} $}\;
			\If{$ v \in \avibfp[v] $}
			{
				\Start $\textit{\ibfpmone{}}(h_v,v)$
			}
		}
	}
	\UponRules{	 
	\Gc{%
		$igc_1(v) \equiv \newline \big( {\ensuremath{\exists \tnsmessage{FINALISED-BLOCK}{\fb} \in {RM}_v }} :$\newline
		\mbox{\hspace{1em}}\mbox{$\BlockHeight{\ExtractBlock{\fb}} = h_v \:\land$}\newline
		\mbox{\hspace{1em}}${isValidFinalisedBlock}(\fb,v)\newline\big)$
	} {	
		\Indp\Indp        
		${chain}_v[h_v] \gets \fb$\;
		\If{$ v \in \avibfp[v] $}{
			\Stop $ \textit{\ibfpmone{}}(h_v,v)$
		}
		$h_v \gets h_v + 1$\;
		\If{$ v \in \avibfp[v] $}
		{
			\Start $\textit{\ibfpmone{}}(h_v,v)$\
		}
	}
	\deleted{$igc_2(v) \equiv \ldots$}
	}
	\caption{\ibftmone{} protocol for node $v$}
	\label{algo:ibftm1-protocol}
\end{algorithm}

\begin{algorithm}[t!]
	\caption{Sequential steps of the $ \textit{\ibfpmone{}}(h,v) $ process (initialisation).}
	\label{algo:ibftm1-block-finalisation-protocol-start-and-init}
	\Expansions{
		\lExpand{\nhmacro[]}{\nmacro[]}
		\lExpand{\avhmacro[]}{\avibfp[]}
		\lExpand{\hproposer[] } {\proposermacro[]}}
	\Functions{
		\lFn{$ {RoundTimerTimeout}(r)$}{$\TimeoutForRoundZero \cdot 2^r$}	
		\lFn{$ \cs{B,v} $}{\Sign{$ \kec{B},v $}}
	}
	\Init{
		\Fn{$ fpInit(h,v) $}
		{
		${lockedBlock}_{h,v} \gets \bot $\;
		$ StartNewRound(0,h,v) $\;
		}
	}
	\Procedures{
		\Proc{$moveToNewRound(r, h, v)$\label{ln:move-to-new-round}}
		{
			$r_{h,v} \gets r$\;
			${roundAlreadyStarted}_{h,v} \gets {false}$\;
			${acceptedBlock}_{h,v} \gets \bot$\label{ln:ibft1:set-accept-pre-prepare-1}\;
		}
		\Proc{$ StartNewRound(r,h,v) $}{
			$moveToNewRound(r, h, v)$\;
			${roundAlreadyStarted}_{h,v} \gets {true}$\;
			${commitSent}_{h,v} \gets {false}$\;
			${finalisedBlockSent}_{h,v} \gets false$\;
			${roundTimerExpiration}_{h,v}[r_{h,v}] \gets now(v) + RoundTimerTimeout(\ensuremath{r_{h,v}})$\;
			\If{$v = \hproposer[] $}{
				\eIf{${lockedBlock}_{h,v} \neq \bot$}{
					$B \gets {lockedBlock}_{h,v} $\;
				}
				{
					$B \gets \CreateNewProposedBlock{\ensuremath{h,v}} $
				}
				\Multicastv{\tsmessage{PRE-PREPARE}{h, r_{h,v}, B}{v}} \To \avhmacro[]\;
			}
		}
		\Proc{$moveToNewRoundAndSendRoundChange(r, h, v)$}
		{
			$moveToNewRound(r, h, v)$\;
			\Multicastv{\tsmessage{ROUND-CHANGE}{h, r_{h,v}}{v}} \To \avhmacro[]\;
		}
	}
\end{algorithm}

In this section we describe a modification to the IBFT protocol to achieve optimal Byzantine-fault-tolerance persistence.
We denote the protocol resulting from this modification as \ibftmonep{} and the block finalisation protocol of \ibftmone{} as \ibfpmone{}.

\cref{algo:ibftm1-protocol} details the sequential steps for the \ibftmonep{} while \crefrange{algo:ibftm1-block-finalisation-protocol-pre-prepare-and-prepare}{algo:ibftm1-block-finalisation-protocol-commit-and-round-change} detail the sequential steps for \ibfpmone{} where yellow background colour is used to highlight the pieces of pseudocode that differ from the original protocol and strikethrough lines are used to indicate pieces of pseudocode that are no longer present in the modified protocol.

\begin{algorithm}[t!]
	
	\caption{Sequential steps of the $ \textit{\ibfpmone{}}(h,v) $ process (PRE-PREPARE and PREPARE).}
	\label{algo:ibftm1-block-finalisation-protocol-pre-prepare-and-prepare}
	
	\UponRules{
		\Gc{$fpgc_1(h,v) \equiv \newline
			\big( \hspace{0.4em}{acceptedBlock}_{h,v} = \bot \:\land\newline
			\mbox{\hspace{0.9em}}\big( \mbox{\ensuremath{\exists\tsmessage{PRE-PREPARE}{h, r_{h,v}, B}{sv } \in {RM}_v :  sv = \hproposer[] \big) }} \newline\big)
			$
		}{
			\Indp\Indp
			\eIf{$ ({lockedBlock}_{h,v} = \bot \lor {lockedBlock}_{h,v} = B) \land \Valid{\ensuremath{B,{chain_v}[h-1]}} $}
			{
				${acceptedBlock}_{h,v} \gets B$\;\label{ln:ibft1:set-accept-pre-prepare-2}
				\Multicastv{\tsmessage{PREPARE}{h, r_{h,v}, \kec{B}}{v}} \To \avhmacro[]\;	
			}
			{
				$moveToNewRoundAndSendRoundChange(r_{h,v}+1, h, v)$\;
			}
		}
		
		\Gc{%
			$fpgc_2(h,v) \equiv \newline
			\big( \hspace{0.4em} {commitSent}_{h,v} = {false} \: \land$\newline 
			\mbox{\hspace{1em}}${{acceptedBlock}_{h,v}}  \neq \bot \: \land$\newline
			\mbox{\hspace{1em}}\SizeOf{$ \{ 
				\tsmessage{PREPARE}{h, r_{h,v}, \kec{{acceptedBlock}_{h,v}}}{sv} \in {RM}_v  \mid \newline
				\mbox{\hspace{28em}}sv \in \avhmacro[]\}$ \newline				    
				\mbox{\hspace{4.3em}}$  $} $ \geq \changed{\ensuremath{\quorumoptimal(\nhmacro[])}}$  \newline $ \big) $
		}{
			\Indp\Indp${lockedBlock}_{h,v} \gets {acceptedBlock}_{h,v}$\;\label{ln:ibftm1:lock1}
		}		
		\Gc{$fpgc_3(v) \equiv \newline
			\bigg( {commitSent}_{h,v} = {false} \:\land $ \newline
			\mbox{\hspace{0.9em}}$ {lockedBlock}_{h,v} \neq \bot \: \land $ \newline
			\mbox{\hspace{0.8em}}$ \ensuremath{\Big(\mbox{\hspace{0.0em}}\big( \exists \tsmessage{PRE-PREPARE}{h, r_{h,v}, {lockedBlock}_{h,v} }{sv} \in {RM}_v  : \mbox{\ensuremath{sv = \hproposer[] \big) \:\lor}}}\newline			
			\mbox{\hspace{0.8em}}\mbox{\hspace{0.6em}}\ensuremath{\big(\exists\tsmessage{PREPARE}{h, r_{h,v}, \kec{{lockedBlock}_{h,v} }}{sv}  \in {RM}_v  : sv \in \avhmacro[]  \big)\Big) } $\newline$ \bigg) $
		}{
			\Indp\Indp
			\Multicastv{}\newline
			{\tsmessage{COMMIT}{h, r_{h,v}, \kec{{lockedBlock}_{h,v}},\cs{{lockedBlock}_{h,v},v}}{v}} \To \avhmacro[]\;\label{ln:ibftm1:send-commit}
			${commitSent}_{h,v} \gets {true}$\;
		}
	}
\end{algorithm}

The summary of list of changes to be applied to obtain the modified protocol follows:
\begin{enumerate}[label=(PM-\arabic{*}),leftmargin=*]
	\item \label{itm:pm-devp2p-mod} The protocol used by IBFT to transmit finalised blocks over the \devptwop{} protocol is modified so that nodes query their peers about the availability of new finalised blocks on a regular basis;
	\item \label{itm:pm-check-signature-for-commit-messages} Only consider Commit messages that have a commit seal signed by the sender of the Commit message.
	Since we assume a digital signature scheme that ensures signature uniqueness, verifying the signature implies verifying that the size of the commit seal is correct.
	\item \label{itm:pm-move-block-verification} Embed the verification that the block included in the accepted Pre-Prepare message allows adding the finalisation proof to it in the series of verifications performed by the \Valid{$\cdot$} function and, therefore, remove the \WellFormedToAddFinalProof{$B$} check from the list of verifications listed at \cref{ln:wellFormedBlock} of \cref{algo:ibft-block-finalisation-protocol-commit-and-round-change}.
	\item \label{itm:pm-remove-failure-in-creating-the-finalisation-proof}  Remove the \textbf{Failure in creating the finalisation proof} condition from the list of conditions that trigger sending Round-Change messages. This corresponds to the condition $  \big( \forall \tsmessage{*}{*, *, *,cs}{*} \in CM:  cs\text{ is of the correct size} \big) $ at \cref{ln:wellFormedBlock} and block \crefrange{ln:fail-finalisation-proof-start}{ln:fail-finalisation-proof-end} of \cref{algo:ibft-block-finalisation-protocol-commit-and-round-change}.
	This means that this condition is never checked and therefore no validator can \unlock[] as consequence of this condition being true.
	\item \label{itm:pm-remove-send-commit-for-old-blocks}Remove the guarded command at \cref{ln:send-commit-for-old-blocks} as it does not provide any value to the protocol.
	While we do not provide a formal proof of such a statement in this paper, its correctness can be informally proved by noticing that (i) as shown in \cref{sec:modification-ibftmonep-achieve-optimal-byzantine-fault-tolerance-persistence}, the modified protocol incorporating the modifications listed here guarantees optimal Byzantine-fault-tolerance persistence and (ii) none of the proposed modifications listed in \cref{sec:solution-to-the-liveness-issue} for addressing the liveness issue of the IBFT protocol involves sending Commit messages in response to Pre-Prepare messages for old blocks, which is essentially the behaviour that the guarded command at \cref{ln:send-commit-for-old-blocks} specifies.	      
	\item \label{itm:pm-replace-quorum-formula} Replace the \quorum{n} function of  the \ibfp{} with the $ \quorumoptimal(n) $ function defined as $ \left \lceil \frac{2n}{3} \right \rceil $.
\end{enumerate}
\begin{algorithm}[t!]
	\caption{Sequential steps of the $ \textit{\ibfpmone{}}(h,v) $ process (COMMIT and ROUND-CHANGE).}
	\label{algo:ibftm1-block-finalisation-protocol-commit-and-round-change}
	
	\UponRules{
		\Gc{$fpgc_4(h,v) \equiv \newline
			\bigg(\hspace{0.2em}{acceptedBlock}_{h,v} \neq \bot \land \newline
			\mbox{\hspace{1em}}\big(\exists \mathit{CM} \in \mathcal{P}({RM}_v) : $\newline
			\mbox{\hspace{2.5em}}\mbox{$ \SizeOf{\ensuremath{CM}} \geq \changed{\ensuremath{\quorumoptimal({\nhmacro[]})}} \:\land$}\vspace{-0.2em}\newline	
			\mbox{\hspace{2.5em}}$ (\forall cm \in CM: cm = \tsmessage{COMMIT}{h, r_{h,v}, \kec{{acceptedBlock}_{h,v}},cs}{sv} \land \vspace{-0.2em}\newline \mbox{\hspace{14.5em}}\changed{\ensuremath{\EthAddressRecover{\ensuremath{\kec{{acceptedBlock}_{h,v}},cs }} = sv  )}}  \: \land $ \newline
			\mbox{\hspace{2.5em}}$ (\forall \smessage{cm}{sv},\smessage{cm'}{sv'} \in CM:  cm \neq cm' \implies sv \neq sv') \newline
			\mbox{\hspace{1.2em}}\big) \:\land $
			\newline
			\mbox{\hspace{1.0em}}$ {finalisedBlockSent}_{h,v} = {false}$ \newline$ \bigg) $
			\label{ln:send-finalise-start}\label{ln:new-send-finalisation-proof-upon-condition}
		}
		{
			\Indp\Indp
				${finalisedBlockSent}_{h,v} \gets {true}$\;
				\EmptyIf{\hspace{-0.3em}\deleted{\textbf{\textup{if}}  $ \big( \forall \tsmessage{*}{*, *, *,cs}{*} \in CM:  cs\text{ is of the correct size} \big) \land$} \vspace{-0.2em}\newline
					\mbox{\hspace{-0.6em}}\deleted{\WellFormedToAddFinalProof{$ {acceptedBlock}_{h,v} $ } \textbf{\textup{then}}}
				}
				{
					${lockedBlock}_{h,v} \gets {acceptedBlock}_{h,v}$\;
					\Let ${finalisationProof} \equiv  \{ cs \mid $
					\mbox{$\tsmessage{*}{*, *, *,cs}{sv} \in CM \} $} \;
					${\fb} \gets \langle {acceptedBlock}_{h,v}, {finalisationProof}\rangle $\;
					\Multicast{$ \tnsmessage{FINALISED-BLOCK}{\fb} $} \To all nodes\;					
				}
				\eEmptyIf{\deleted{\textbf{\textup{else}}}}
				{
					\deleted{$moveToNewRoundAndSendRoundChange(r_{h,v}+1, h, v)$}\vspace{-1.5em}\;
					\deleted{${lockedBlock}_{h,v} \gets \bot$}\;			
				}				
			
			
		}
		
		\Gc{$fpgc_5(h,v) \equiv \newline
			\big( \exists r_{rc}: $
			\mbox{\hspace{1em}}\mbox{\SizeOf{$ \{ \tsmessage{ROUND-CHANGE}{h, r_{rc}}{sv} \in {RM}_v \mid sv \in \avhmacro[] \} $}
				$\geq f(\nhmacro[])+1 \; \land$} \newline$  $
			\mbox{\hspace{1em}}$r_{rc} > r_{h,v}$\newline$ \big) $} {
			
			\Indp\Indp$moveToNewRoundAndSendRoundChange(r_{rc}, h, v)$\;
		}
		\Gc{$ fpgc_6(h,v) \equiv \newline $
			$ \big( now(v) \geq {roundTimerExpiration}_{h,v}[r_{h,v}] \big)$} {
			\Indp\Indp$moveToNewRoundAndSendRoundChange(r_{h,v}+1, h, v)$\;
		}
		\Gc{$ fpgc_7(h,v) \equiv \newline
			\Big( \exists r_{rc}: $\newline 
			\mbox{\hspace{1em}}\SizeOf{$ \{ \tsmessage{ROUND-CHANGE}{h, r_{rc}}{sv} \in {RM}_v \mid sv \in \avhmacro[] $\newline
				\mbox{\hspace{4.2em}}}
			$\geq \changed{\ensuremath{\quorumoptimal({\nhmacro[]})}} \:\land$ \newline  
			\mbox{\hspace{1em}}$\big( r_{rc} > {r}_{h,v} \lor 
			(r_{rc} = {r}_{h,v} \land 
			{roundAlreadyStarted}_{h,v} = {false})\big)$\newline$ \Big) $
			}{
			\Indp\Indp$StartNewRound(r_{rc},h,v)$\;
		}
	}
\end{algorithm}

Modification \ref{itm:pm-devp2p-mod} is applied to ensure condition (ii) of the Persistence property (see \cref{def:robustnes}) provided that the \ibfp{} guarantees $t$-Byzantine-fault-tolerant weak-liveness. Modifications \ref{itm:pm-check-signature-for-commit-messages}, \ref{itm:pm-move-block-verification} and \ref{itm:pm-remove-failure-in-creating-the-finalisation-proof} are applied to achieve Byzantine-fault-tolerant persistence. 
Modification \ref{itm:pm-remove-send-commit-for-old-blocks} removes unnecessary protocol behaviours.
Modification \ref{itm:pm-replace-quorum-formula} is applied to achieve optimal Byzantine-fault-tolerance persistence.


\subsection{Persistence Analysis of the  \ibftmonep{}}\label{sec:safety-analysis-of-the-first-solution}

In this section we analyse the persistence property of the \ibftmonep{} and show that the \ibftmonep{} guarantees optimal Byzantine-fault-tolerant persistence threshold.

We denote the total number of validators for the $h$-th instance of the \ibfpmone{} with $n_h$.
This definition is however meaningful only if all honest nodes agree on the same set of validators, thus we will add this assumption to all the Lemmas where $ n_h $ is used.

Any proof omitted in this section can be found in \ref{ap:proofs-for-safety-analysis-of-the-first-solution}.

\begin{lemma}\label{lem:invariant-for-locked-block}
	For all $ gc \in \{ fpgc_1(h,v),\ldots,fpgc_7(h,v)\}$ executed after the first execution of either $ fpgc_2(h,v) $ or $ fpgc_4(h,v) $ the following invariant, expressed as a Hoare triple, holds for some $ B \neq \bot $
	\begin{equation*}
		\{lockedBlock_{h,v}=B\}gc\{lockedBlock_{h,v}=B\}
	\end{equation*}
\end{lemma}
\begin{proof} 	We will show that the Lemma holds for $ fpgc_2(h,v) $.
	Exactly the same reasoning can be applied to prove that it also holds for $ fpgc_4(h,v) $.
	
	$ fpgc_2(h,v) $ can only be executed if $ acceptedBlock_{h,v} \neq \bot $.
	Since, $ fpgc_2(h,v)  $ assigns $acceptedBlock_{h,v}  $ to $ lockedBlock_{h,v} $, then $ lockedBlock_{h,v} = acceptedBlock_{h,v} = B \neq \bot $ for some $ B\neq \bot $ after the execution of $ fpgc_2(h,v) $.
	The only guarded commands that can change the value of $ lockedBlock_{h,v} $ are $ fpgc_2(h,v) $ and $ fpgc_4(h,v) $.
	The execution of both guarded commands preserves the invariant stated in the Lemma if $ acceptedBlock_{h,v} = B $ is true before their execution.
	$ acceptedBlock_{h,v} = B $ can only be changed by  the execution of $ fpgc_1(h,v) $ which is however disabled after the first execution of $ fpgc_2(h,v) $ as the guard of $ fpgc_1(h,v) $ requires $  acceptedBlock_{h,v} = \bot $.	
\end{proof}

\begin{lemma}\label{lem:no-change-lock-in-future-block-events}
	If an honest validator $v$ \lock[s] on block $B$ while executing the $ h $-th instance of the \ibfp{}, then $v$ will be \lock[ed] on block $B$ for rest of the execution of the $ h $-th instance of the \ibfp{}.
\end{lemma}
\begin{proof}
	The only sequential steps where validator $ v $ can \lock[] on a block $ B $, i.e. $ lockedBlock_{h,v} \gets B $, is in the execution of either $ fpgc_2(h,v) $ or $ fpgc_4(h,v) $.
	This observation and \cref{lem:invariant-for-locked-block} imply the Lemma.
\end{proof}

\begin{lemma} \label{lem:no-relocking}
	If an honest validator $v$ is \lock[ed] on block $B$ while in round $r$ of the $h$-th instance of the \ibfpmone{}, then $v$ will be \lock[ed] on block $B$ for any round $r'\geq r$ of the same instance $h$.
\end{lemma}
\begin{proof}
	After the initialisation of instance $ h $ validator $ v $ is not \lock[ed] on any block.
	Let $ r' $ be the first round where $ v $ \lock[s] on block $ B $.
	\cref{lem:no-change-lock-in-future-block-events} and the observation that round number $r_{h,v}$ increments monotonically prove the Lemma.
\end{proof}

\begin{lemma}\label{lem:safety-solution-forward}
	For any $n_h \geq 1$, provided that the number of Byzantine validators is $\leq f(n_h)$, the intersection of any two sets of \quorumoptimal($ n_h $) validators is guaranteed to include an honest validator.
\end{lemma}

\begin{lemma}\label{lem:if-sent-commit-then-locked}
\begin{sloppypar}		If an honest validator $v$ sends a Commit message for block $B$ (i.e \linebreak \tsmessage{COMMIT}{h, r, \kec{B},\cs{B,v}}{v}) while running the $ h $-th instance of the \ibfpmone{}, then (i) validator $v$ is \lock[ed] on $B$ and (ii) the height of block $ B $ is $ h $, i.e. $ \BlockHeight{B} =h$.
\end{sloppypar}
\end{lemma}
\begin{proof}
	The only line of the pseudocode of \crefrange{algo:ibftm1-protocol}{algo:ibftm1-block-finalisation-protocol-commit-and-round-change} where a Commit message is sent is \cref{ln:ibftm1:send-commit} and this line only sends Commit messages for blocks on which validator $ v $ is \lock[ed].
	This proves condition (i) of the Lemma.

	The only sequential steps where node $ v $ can \lock[] on a block are $ fpgc_2(h,v) $ and $ fpgc_4(h,v) $.
	In either case, $ lockedBlock_{h,v} $ is set to $ acceptedBlock_{h,v} $ and either guarded command can be executed only if $ acceptedBlock_{h,v} \neq \bot $.
	The only sequential step where $ acceptedBlock_{h,v} $ is set is $ fpgc_1(h,v) $ where $ acceptedBlock_{h,v} $ is set to $ B $ provided that   \Valid{\ensuremath{B,{chain_v}[h-1]}} is true.
	According to equation (50) of the Ethereum Yellow Paper \cite{yellowpaper}, \mbox{\Valid{\ensuremath{B,B_{parent}}}} can be true only if $ \BlockHeight{\ensuremath{B}} = \BlockHeight{\ensuremath{B_{parent}}}+1 $.
	This concludes the proof of condition (ii).
\end{proof}

\begin{lemma}\label{lem:safety-forward}
	For any instance $ h $ of the \ibfpmone{} such that all honest nodes agree on the same set of validators and the number of Byzantine validators is $ \leq f(n_h) $,	
	no two valid finalised blocks for the same height $h$ including different blocks can be produced.
\end{lemma}

\begin{proof}
	A block for height $h$ is considered finalised only if it contains a valid finalisation proof.
	A finalisation proof is valid only if it contains at least $\quorumoptimal(n_h)$ commit seals signed by different validators in \avhmacro[].
	The proof is by contradiction. 
	Assume that (a) there exist two finalised blocks $\fb$ and ${\fb{}}'$ including block $B$ and $B'$ respectively, with $B \neq B'$.
	Since, according to \cref{lem:safety-solution-forward}, the intersection of any two sets of $\quorumoptimal(n_h)$ validators is guaranteed to include an honest validator and we assume unforgeability of digital signature, assumption (a) implies that there exists one honest validator $v$ that produced both a commit seal for block $ B $ and commit seal for block  $B'$.
	This, in turn, implies that validator $v$ sent a Commit message with block hash matching $\kec{B}$ and a Commit message with block hash matching $\kec{B'}$ as commit seals are only included in Commit messages.
	However, according to condition (i) of \cref{lem:if-sent-commit-then-locked}, honest validators only send Commit messages for the block that they are \lock[ed] on. 
	Therefore, since we assume that \kec{\cdot} is a collision-resistant hash function (i.e $B \neq B' \rightarrow \kec{B} \neq \kec{B'}$ with high probability ), for assumption (a) to be satisfied $v$ must have \lock[ed] on both block $B$ and $B'$.
	This is in clear contradiction with \cref{lem:no-relocking}.
\end{proof}

\begin{lemma}\label{lem:no-valid-finalised-block-with-heigh-h'-can-be-produce-by-instance-h-of-ibfpone}
	For any instance $ h $ of the \ibfpmone{} such that all honest nodes agree on the same set of validators and the number of Byzantine validators is $ \leq f(n_h) $,	
	only valid finalised blocks for height $ h $ can be produced.
\end{lemma}
\begin{proof}
	The proof is by contradiction.
	Let us assume that there exists a valid finalised block \fb{} for height $ h' $ produced by the $ h $-th instance of the \ibfpmone{}, with $ h' \neq h $.
	Since \fb{} is a valid finalised block produced by instance $ h $ of the \ibfpmone{}, it must contain $ \quorumoptimal(n_h) $ commit seals.
	Given that $ \quorumoptimal(n_h) > f(n_h)$ for any $ n_h \geq 1 $, then \fb{} must contain at least one commit seal signed by an honest validator.
	Since commit seals are only included in Commit messages, this in turn implies that at least one honest validator sent a Commit message for a block with height $ h' $ while running the $ h $-th instance of the \ibfpmone{}.
	This is in contradiction with condition (ii) of \cref{lem:if-sent-commit-then-locked}.
\end{proof}

\begin{lemma}\label{lem:safety-iff-intersection}
	For any instance $ h $ of the \ibfpmone{} such that all honest nodes agree on the same set of validators, then Byzantine-fault-tolerant safety is guaranteed if and only if the number of Byzantine validators is $ \leq f(n_h) $.
\end{lemma}

\begin{proof}
	The forward direction of the Lemma is proved by \cref{lem:safety-forward,lem:no-valid-finalised-block-with-heigh-h'-can-be-produce-by-instance-h-of-ibfpone} while the reverse direction is a direct consequence of $ f(n_h) $ being the maximum number of Byzantine validators that any consensus protocol operating in an eventually synchronous network can withstand \cite{Dwork:1988:CPP:42282.42283} (see also \cref{sec:system-model}).
\end{proof}




\begin{theorem}
	The \ibftmonep{} guarantees optimal Byzantine-fault-tolerance persistence threshold. 
\end{theorem}
\begin{proof}
\begin{sloppypar}
		Follows from \cref{def:optimal-Byzantine-fault-tolerance-persistence} of optimal Byzantine-fault-tolerance persistence threshold and \cref{lem:safety-iff-intersection}.
\end{sloppypar}
\end{proof}

\subsection{Liveness impact of modification \ibftmone{}}\label{sec:liveness-impact-of-modification-ibftmone}

In this section we show that modification \ibftmone{} does not affect the weak-liveness of the \ibfp{}.

\begin{lemma}\label{lem:qopt+f-leq-n}
	$ \forall n \geq 0 : \quorumoptimal(n) + f(n)\leq n $
\end{lemma}
The proof of the previous Lemma can be found in \ref{ap:proof-for-liveness-impact-of-modification-ibftmone}.

\begin{lemma}\label{lem:liveness-not-affected-by-replacement-of-q-with-qopt}
	The liveness property of the IBFT protocol is not affected if \quorum{n} is replaced by $\quorumoptimal(n)$ provided that the number of Byzantine validators for any instance $h$ of the \ibftmonep{} is $ \leq f(n_h)$.
\end{lemma}
\begin{proof}
	\Cref{lem:qopt+f-leq-n} implies that for any round of the $h$-th instance of the \ibftmonep{}, there exists at least $ \quorumoptimal(n_h) $ honest validators. 
	This, in turn, implies that once GST is reached, if all $n_h$ honest validators are \unlock[ed] and the proposer for the current round is honest, then Byzantine validators cannot prevent honest validators from receiving $\quorumoptimal(n_h)$ Prepare messages and $\quorumoptimal(n_h)$ Commit messages for the block proposed by the block proposer.
\end{proof}

\begin{theorem}\label{lem:liveness-not-affected-by-ibftmone}
	The modification \ibftmone{} does not affect the weak-liveness property of the \ibfp{}.
\end{theorem}
\begin{proof}
	The proof is given informally by contending that each of the modifications listed in \cref{sec:modification-ibftmonep-achieve-optimal-byzantine-fault-tolerance-persistence} does not affect the weak-liveness property  of the \ibfp{}:
\begin{itemize}
	\item Modification \ref{itm:pm-devp2p-mod}: this modification only deals with blocks that have already been finalised.
	
	\item Modification \ref{itm:pm-check-signature-for-commit-messages}: honest validators only send Commit messages with a commit seal signed by them and, according the \cref{lem:qopt+f-leq-n}, there are at least $ \quorumoptimal(n_h) $ honest validators.
	
	\item Modification \ref{itm:pm-move-block-verification}: honest proposers send Pre-Prepare messages that include a block well formed for adding a finalisation proof to it and according the \cref{lem:qopt+f-leq-n}, there are at least $ \quorumoptimal(n_h) $ honest validators.
	
	\item Modification \ref{itm:pm-remove-failure-in-creating-the-finalisation-proof}:  (i) by Modification \ref{itm:pm-move-block-verification} only blocks that are well formed for adding a finalisation proof to them are ever going to be accepted by an honest validator and (ii) 
	\cref{ln:unlocks} in \cref{algo:ibft-block-finalisation-protocol-commit-and-round-change}
	is only executed if a block included in a Pre-Prepare message has been accepted.
	
	\item Modification \ref{itm:pm-remove-send-commit-for-old-blocks}:  this modification only deals with sending Commit messages upon reception of a Pre-Prepare message for an old instance of the \ibfp{} for which the node receiving the Pre-Prepare message has already received a valid finalised block, i.e. instances of the \ibfp{} that have already achieved weak-liveness.
	
	\item Modification \ref{itm:pm-replace-quorum-formula}: see \Cref{lem:liveness-not-affected-by-replacement-of-q-with-qopt}.
\end{itemize}
\end{proof}

\section{Liveness Analysis}\label{sec:liveness-analysis}

In this section we analyse the liveness property of the IBFT and \ibftmone{} protocols. 
Specifically, we show that neither the IBFT protocol nor the \ibftmone{} protocol guarantee liveness when operating in an eventually synchronous network model when either a Byzantine or a fail-stop failure model are considered, and that the IBFT-M1 protocol may be modified in two ways to guarantee liveness.

\subsection{Analysis of the IBFT and \ibftmone{} protocols}\label{sec:analysis-of-the-ibft-and-ibftmone-protocols}

Any proof omitted in this section can be found in \ref{ap:proofs-for-analysis-of-the-ibft-and-ibftmone-protocols}.

\begin{lemma}\label{lem:liveness-weak-liveness-relationship}
	The IBFT protocol (\ibftmonep{}, resp.) guarantees $t$-Byzantine-fault-tolerant liveness only if the \ibfp{} (\ibfpmone{}, resp.) guarantees $t$-Byzantine-fault-tolerant weak-liveness.
\end{lemma}
\begin{proof}
	If the \ibfp{} (\ibfpmone{}, resp.)  does not guarantee $t$-Byzantine-fault-tolerant weak-liveness, then there could exist an instance $h$ of the \ibfp{} (\ibfpmone{}, resp.) such that no finalised block is ever produced.
	This implies that no finalised block with height $\geq h$ will ever be added to the local blockchain of any honest node as honest nodes must add a finalised block with height $h$  to their local blockchain before being able to add any finalised block with  height higher than $h$.
	Hence, any transaction sent to all honest validators that is not already included in a finalised block with height $< h$ will never be added to the local blockchain of any honest validator.
\end{proof}


\begin{lemma}\label{lem:n-q<q}
	$ \forall n\geq 4 \land n\neq 6: n - \quorum{n} < \quorum{n} \label{eq:n-quorum<quorum} $ 
\end{lemma}

%

\begin{lemma}\label{lem:n-qopt<qopt}
	$ \forall n\geq 0 : n - \quorumoptimal(n) < \quorumoptimal(n) $
\end{lemma}

\begin{lemma}\label{lem:no-liveness}
	For any $h$ instance of the block finalisation protocol where $n_h \geq 4$, neither the \ibfp{} nor the \ibfpmone{} guarantee weak-liveness in the presence of a single faulty node.
	This is true even if the only type of node failure considered is fail-stop.
\end{lemma}
\begin{proof}
	The proof is by counterexample.
	This is achieved by providing a possible sequence of events, that may occur before GST, which causes the protocol to be unable to create any new finalised block, even once GST is reached.
	We consider events occurring during the $h$-th instance of the \ibfp{}and \ibftmonep{}, and assume that one of the $n_h$ validators for the $h$-th instance is faulty and therefore can stop communicating and never restart.
	Also, let us assume that all of the events outlined below happen before GST.
	We consider three possible cases.
	\begin{description}[style=nextline]
		\item[Case 1: \textnormal{\itshape\ibfp{} with $n_h \geq 4 \land n_h \neq 6$.}]
				Let  $W$ be a subset of $\quorum{n_h}$ validators containing the proposer for round 1, $p_1$, and the faulty validator.
				Let $V$ be the complement of $W$. 
				The size of $V$ is therefore $n_h - \quorum{n_h}$ which, according to \cref{lem:n-q<q}, is $< \quorum{n_h}$ for $n_h \geq 4 \land n_h \neq 6$.
				The immediate consequence of the definition above is that all validators in $V$ are honest. .
				\begin{enumerate}
					\item 	The proposer for round 0, $p_0$, multicasts a Pre-Prepare message for block $B$ \linebreak \tsmessage{PRE-PREPARE}{h, 0, B}{p_0}.
					\item All validators reply with a Prepare message for block $ B $ \tnsmessage{PREPARE}{h, 0, \kec{B}}.
					\item All validators in $ V $ receive the Pre-Prepare message for block $ B $ and all of the Prepare messages. 
					Therefore, all validators in $ V $ \lock[] on block $ B $.
					\item The round-timer expires before any other validator, except for those in set $ V $, receives any of the Prepare messages.
					Note that since $\| V \| < \quorum{n_h}$ and  only validators in $V$ may send a Commit messages for block $ B $, no honest validator will ever produce a finalised block including block $ B $ due to the Commit messages sent in round 0.
					\item All validators send a Round-Change message \tnsmessage{ROUND-CHANGE}{h, 1} to move to round 1.
					\item All validators receive $\quorum{n_h}$ Round-Change messages and move to round 1.
					\item  $ p_1 $ sends a Pre-Prepare message for block $B'$ \tsmessage{PRE-PREPARE}{h, 1, B'}{p_1}, with $B' \neq B$. 
					This is possible as $p_1$ is not \lock[ed] on any block since it is an element of set $W$ which is disjoint from set $ V $.
					\item All validators included in set $ V $ are locked  on $ B $, therefore when they receive the Pre-Prepare message sent by $ p_1 $ they reply with a Round-Change \tnsmessage{ROUND-CHANGE}{h, 2} to move to round 2.
					\item All of the validators in set $ W $ respond to the Pre-Prepare for block $B'$ that they received with a Prepare message for block $B'$ and round 1 \tnsmessage{PREPARE}{h, 1, \kec{B'}}. 
					However, one of these validators is faulty and stops just after sending a Prepare message and before sending any Commit message.
					Set $W$ now contains only $\quorum{n_h}-1 < \quorum{n_h}$ active validators.
					\item All honest validators in set $W$ receive all the Prepare messages for block $B'$ sent by all validators in $W$ including the Prepare message sent by the faulty validator before stopping. 
					Since $\| W \| \geq \quorum{n_h}$, all honest validators in set $W$ \lock[] on block $B'$.
				\end{enumerate}
				Let $W_{active}$ be the subset of active validators in $W$.
				By the end of the last step of the sequence of events presented above, all validators, except the faulty one that stopped, are locked on a block.
				Specifically, validators in $V$ are locked on $B$ and validators in $W_{active}$ are locked on $B'$.
				In the \ibfp{} honest validators can release the \lock[] only if they receive \quorum{n_h} Commit messages with height, round and block matching the Pre-Prepare message, while in the \ibftmonep{} honest validators never release a lock.
				Since (i) the union of sets $W_{active}$ and $V$ corresponds to the entire set of all honest and active validators, (ii) both sets $W_{active}$ and $V$ have size $< \quorum{n}$ and (iii) honest validators that are locked on a block only send messages for this block, then no validator will ever receive $\quorum{n_h}$ Commit messages for the same block and therefore no honest validator will either \unlock[] or create a finalised block.\\
		\item[Case 2: \textnormal{\itshape\ibfp{} with $n_h = 6$.}]
			Compared to Case 1, we assume that (i) set $W$ also contains the proposer for round 2, $p_2$, (ii) set $V$ is defined to be a subset of $\quorum{n_h}-1$ validators with no intersection with set $W$ and (iii) we define set $Z$ to be the complement of the union of $W$ and $V$.
			The three sets are therefore disjoint sets, and set $Z$ has size $\|Z\|=n_h-(2\cdot\quorum{n_h}+1)$ which for $n_h=6$ corresponds to  $1 < \quorum{n_h}$.
			The immediate consequence of the definition above is that all validators in $V$ and $Z$ are honest. 
			The sequence of events happening before GST is similar to the one presented for Case 1 with the addition of events \labelcref{item:start} to \labelcref{item:end}.			
			\begin{enumerate}
				\item 	The proposer for round 0, $p_0$, multicasts a Pre-Prepare message for block $B$ \linebreak \tsmessage{PRE-PREPARE}{h, 0, B}{p_0} to all validators (including itself).
				\item All validators reply with a Prepare message for block $ B $ \tnsmessage{PREPARE}{h, 0, \kec{B}}.
				\item All validators in $ V $ receive the Pre-Prepare message for block $ B $ and all of the Prepare messages. 
				Therefore, all validators in $ V $ \lock[] on block $ B $.
				\item The round-timer expires before any other validator, except for those in set $ V $, receives any of the Prepare messages.
				Note that since $\| V \| < \quorum{n_h}$ and  only validators in $V$ may send a Commit message for block $ B $, no honest validator will ever produce a finalised block including block $ B $ due to the Commit messages sent in round 0.
				\item All validators send a Round-Change message \tnsmessage{ROUND-CHANGE}{h, 1} to move to round 1.
				\item All validators receive $\quorum{n_h}$ Round-Change messages and move to round 1.
				\item  $ p_1 $ sends a Pre-Prepare message for block $B'$ \tsmessage{PRE-PREPARE}{h, 1, B'}{p_1}, with $B' \neq B$. 
				This is possible as $p_1$ is not \lock[ed] on any block since it is an element of set $W$ which is disjoint from set $ V $.
				\item All validators included in set $ V $ are locked  on $ B $. Therefore when they receive the Pre-Prepare message sent by $ p_1 $ they reply with a Round-Change \tnsmessage{ROUND-CHANGE}{h, 2} to move to round 2.			
				\item \label{item:start} All of the validators in sets $ Z $ and $W$ respond to the Pre-Prepare for block $B'$ that they received with a Prepare message for block $B'$ and round 1 \tnsmessage{PREPARE}{h, 1, \kec{B'}}.
				\item All validators in $ Z $ receive the Pre-Prepare message for block $ B' $ and all of the Prepare messages sent by validators in sets $Z$ and $W$.
				Since $\|Z \cup W\| \geq \|W\| =\quorum{n_h}$,  all validators in $ Z $ \lock[] on block $ B' $.
				
				\item The round-timer for round 1 expires before any other validator, except for those in sets $V$ and $Z$, receives any of the Prepare messages.
				Note that since $\| Z \| < \quorum{n_h}$ and  only validators in $Z$ may send a Commit message for block $ B' $, no honest validator will ever produce a finalised block including block $ B' $  due to the Commit messages sent in round 1.
				\item All validators send a Round-Change message  \tnsmessage{ROUND-CHANGE}{h, 2} to move to round 2.
				\item All validators receive $\quorum{n_h}$ Round-Change messages and move to round 2
				\item  \label{item:end} $ p_2 $ sends a Pre-Prepare message for block $B''$ \tsmessage{PRE-PREPARE}{h, 2, B''}{p_2}, with $B'' \neq B' \neq B$.
				This is possible as $p_2$ is not \lock[ed] on any block since it is an element of set $W$ which is disjoint from the union of sets $V$ and $Z$. 
				
				\item All validators included in sets $ V $ and $Z$ are locked  on $ B $ and $B'$, respectively. Therefore when they receive the Pre-Prepare message sent by $ p_2 $ they reply with a Round-Change to move to round 3.	
				All of the validators in set $W$ respond with a Prepare message for block $B''$ and round 2.
				However, one of these validators is faulty and stops just after sending the Prepare message and before sending the related Commit message.
				Set $W$ now contains only $\quorum{n_h}-1 < \quorum{n_h}$ active validators.
				\item All honest validators in set $W$ receive all the Prepare messages for block $B''$ sent by all validators in $ W$ including the Prepare message sent by the faulty validator before stopping. 
				Since $\| W \| = \quorum{n_h}$, all honest validators in set $W$ \lock[] on block $B''$.
				
			\end{enumerate}
			Let $W_{active}$ be the subset of active validators in $W$.
			By the end of the last step of the sequence of events presented above, all validators, except the faulty one that stopped, are locked on a block.
			Specifically, validators in $V$ are locked on $B$, validators in $Z$ are locked on $B'$ and validators in $W_{active}$ are locked on $B''$.
			In the \ibfp{} honest validators can release the \lock[] only if they receive \quorum{n_h} Commit messages with height, round and block matching the Pre-Prepare message, while in the \ibftmonep{} honest validators never release a lock.
			Since (i) the union of sets $W_{active}$, $V$ and $Z$ corresponds to the entire set of all honest and active validators, (ii) all sets $W_{active}$, $V$ and $Z$ have size $< \quorum{n}$ and (iii) honest validators that are locked on a block only send messages for this block, then no validator will ever receive $\quorum{n_h}$ Commit messages for the same block and therefore no honest validator will either \unlock[] or create a finalised block.\\				
	\item[Case 3: \textnormal{\itshape\ibfpmone{} with $n_h \geq 2$.}]
			This case corresponds to Case 1 with \quorum{n_h} replaced by $\quorumoptimal(n_h)$.
			In this case, the size of $V$ is $n_h - \quorumoptimal(n_h)$ which, according to \cref{lem:n-qopt<qopt}, is $< \quorumoptimal(n_h)$ for $n_h \geq 1$. 
			The resulting sequence of events leads the \ibftmonep{} to a state where it will be unable to create any new finalized block even once GST is reached.
			
	\end{description}

	\paragraph{Remark}
	For $n_h = 2 \lor n_h = 3$, the \ibfp{} guarantees pathological weak-liveness as in this case $\quorum{n_h}=1$ and therefore each honest validator can create a valid finalised block without requiring interaction with any other validator.
\end{proof}

\begin{theorem}
	When operating in an eventually synchronous network model, neither the IBFT protocol nor the \ibftmonep{} guarantees liveness. This is true even if the only type of node failure considered is fail-stop.
\end{theorem}

\begin{proof}
	Direct consequence of \cref{lem:liveness-weak-liveness-relationship,lem:no-liveness} and the fact that we are interested only in the cases where Byzantine-fault-tolerance can be guaranteed, i.e the number of validators is $\geq 4$.
\end{proof}

\subsection{Modification \ibftmtwop{}: Ensure liveness}\label{sec:solution-to-the-liveness-issue}
In this section we explore two options of how the IBFT and \ibftmone{} protocols can be modified to guarantee liveness when operating in an eventually synchronous network.

\subsubsection{PBFT-like solution}
The first option that we present is based on the original PBFT protocol \cite{Castro:1999:PBF:296806.296824}.
The key change to IBFT is to remove the locking mechanism and ensure safety between round changes via a round change sub-protocol similar to the one in the PBFT protocol.
The key modifications can be summarised as follows:
\begin{enumerate}[label=(S1.M-\arabic{*}),leftmargin=*]
	\item Remove the locking logic. \label{en:remove-lock}
	\item \label{en:add-prepared-certificate} 	
	Add a \emph{Prepared Certificate} to the Round-Change message.
	A Prepared Certificate includes a set of $ \quorumoptimal(n_h) $ Prepare messages received for the same round of instance $ h $ of the block finalisation protocol.
	If a validator has  received multiple set of $ \quorumoptimal(n_h) $ Prepare messages for different rounds of instance  $ h $ of the block finalisation protocol, then its Prepared Certificate only contains the set of  $ \quorumoptimal(n_h) $ Prepare messages for the highest round number.
	\item \label{en:add-new-round-message} Add a New-Round message to be sent by the proposer of the new round once $\quorumoptimal(n_h)$ valid Round-Change messages for the same new round are received by the proposer.\\
	The content of the New-Round message is similar to the content of the New-View message of the PBFT protocol, but with the following differences:
	\begin{enumerate}[label=(S1.M-\arabic{enumi}.\arabic{*}),leftmargin=*]
		\item Since IBFT decides on a block at a time, only one Pre-Prepare message is included in the New-Round message.
		\item The Pre-Prepare message included in the New-Round message always contains a valid block. 
		If the set of Round-Change messages included in the New-Round message contains at least one valid Prepared Certificate, then the block included in the Pre-Prepare message must match the block that has a valid Prepared Certificate with the highest round number.
		Otherwise, the Pre-Prepare message can contain any valid  block.
		This is in contrast with PBFT where the Pre-Prepare messages included in the New-View message may contain \emph{null} requests.
	\end{enumerate}
	\item In contrast to PBFT, no checkpoint is required here as the finalisation proof added to each IBFT block serves as checkpoint.
\end{enumerate}

Modification~\ref{en:remove-lock} obviously solves the liveness issue, but if introduced in isolation, it would make the protocol unsafe in the case that a round change occurs.
Modification~\ref{en:add-prepared-certificate} and \ref{en:add-new-round-message} ensure safety without compromising liveness. 

A few protocol optimisations are required to this solution to reduce the size of Round-Change and New-Round messages.
One of the potential solutions for reducing the size of the New-Round messages is replacing the block included in the signed Pre-Prepare messages with a digest of the block and piggyback the full block to the signed Pre-Prepare, Round-Change and New-Round messages.
Other solutions are available but their discussion is outside the scope of this work.

The correctness proof of this solution is quite extensive and therefore it will be presented as part of a separate body of work.
 
\subsubsection{Tendermint-like solution}
In contrast, the second option looks at implementing some of the concepts from the Tendermint variant of the PBFT protocol \cite{buchman2018latest}.
The availability of this type of solution stems from a private conversation with Clearmatics.
The key idea here is borrowing the \emph{relocking} mechanism from the Tendermint protocol \cite{buchman2018latest}.
The list of modifications that should be applied to IBFT to implement this solution are summarised below:
\begin{enumerate}[label=(S2.M-\arabic{*}),leftmargin=*]
	\item Add the \emph{locked round} value to the Pre-Prepare message. The locked round value corresponds to the latest round number where a validator \lock[ed] on the locked block.
	\item Allow relocking if a validator receives $\quorumoptimal(n)$ Prepare messages for round $r$ and one Pre-Prepare message with locked round equal to $r$ provided that $r$  is higher than the current round.
\end{enumerate}

\subsubsection{Comparison between the two solutions}
The two proposed solutions differ in behaviour and performance only when either the proposer is Byzantine or the network delay is longer than the duration of the round timer for round $0$.
If the proposer for round  $0$ is honest and all messages are delivered within the duration of the round timer for round $0$, then the two solutions behave identically from a practical perspective.
Performance wise, the main differences between the two solutions can be summarised as follows:
\begin{itemize}
	\item The PBFT-like solution ensures that after GST, if the proposer of the current round is honest, then every honest node decides within that round. In contrast, the Tendermint-like solution may require $n_h-1$ round changes before an honest validator reaches a decision where $n_h$ is the number of validators for the generic $h$-th instance of the \ibfp{};
	\item The PBFT-like solution can be extended to achieve block finalisation in two communication phases only (Pre-Prepare and Prepare) in the optimal case that no Byzantine validators are present and the network delay is longer than the duration of the round timer for round $0$.
	This can be achieved by adapting the protocol for very fast learning presented by Dutta et al \cite{best-case-complexity-asynch-byz-consnesus} to work within the definition of the IBFT protocol;
	\item The Tendermint-like solution requires lower overall bandwidth than the PBFT-like solution when the protocol moves to a new round.

\end{itemize}

\section{Conclusions and Further Work}

To our knowledge, this work represents the first analysis of the robustness of the IBFT consensus protocol. Our work analysed the persistence and liveness properties of the protocol, and discovered issues in both. 

Persistence guarantees blockchain consistency and immutability amongst all honest nodes. Specifically, persistence ensures that if an honest node adds transaction $ T $ in position $ i $ of its local transaction ledger, then (i) $ T $ is the only transaction that can ever be added in position $i$ by any other honest node, (ii) $T$ will eventually be added to the local transaction ledger of any other honest node. We showed that the IBFT protocol does not guarantee persistence under an eventually synchronous network model.

Liveness ensures that a transaction sent to all honest validators will eventually be included in the blockchain. We showed that the IBFT protocol does not guarantee liveness under an eventually synchronous network model even if the only failure mode considered is fail-stop.

%

We briefly proposed five modifications to the IBFT protocol so that persistence is ensured when operating in eventually synchronous networks. We also proposed two possible ways to modify the IBFT protocol to ensure liveness when operating in eventually synchronous networks: modification of IBFT based on lessons from the Practical Byzantine Fault Tolerance (PBFT) consensus protocol, or modification based on lessons from the Tendermint consensus protocol. Such proposals would need to be both fully developed and formally analysed to be properly assessed.

\section*{Acknowledgements}

The authors wish to thank Rob Dawson for reviewing multiple versions of this document, providing insightful comments and helping to improve the overall readability, Dr. Sandra Johnson for verifying its overall soundness, fixing several typos and further improving the document readability, Franck Cassez for reviewing the document, suggesting to use process algebra notation and identifying a few minor issues with the definitions and proofs, Trent Mohay for reviewing the description of the IBFT protocol in one of the initial versions of the document and identifying a few errors and omissions, Nicolas Liochon and Ben Edgington for providing insightful comments upon review of earlier versions of the work, and Gautam Botrel for verifying the pseudocode description against the IBFT protocol implementation.

Acknowledgements also go to Clearmatics, who suggested the Tendermint-like solution and pointed out that one of our earlier versions of the protocol description failed to specify how validators respond to old Pre-Prepare messages.

\bibliographystyle{elsarticle-num} 
\bibliography{bibliography}

\appendix
\section{Proofs for \Cref{sec:safety-analysis}}\label{ap:proofs-for-persistence analysis}

\newcommand{\setlemmanumberto}[1]{
	\setcounter{lemma}{\getrefnumber{#1}}
	\addtocounter{lemma}{-1}
}
\setlemmanumberto{lem:transaction-t-uniquiness}
\begin{lemma}
	A transaction $ T $ cannot appear in two different positions of the local blockchain of an honest node. 
\end{lemma}

\begin{proof}
	An honest node, before adding a finalised block to its local blockchain, checks the block included in the finalised block for validity (\cref{ln:check-block-validity} in \cref{algo:ibft-protocol}).
	As detailed in \Cref{sec:protocol-description}, block validity, as far as transactions are concerned, is verified as specified in the Ethereum Yellow Paper \cite{yellowpaper}.
	Equation (58) of the Ethereum Yellow Paper mandates that a transaction is valid only if the transaction nonce matches the expected nonce of the transaction sender \cite{yellowpaper}. 
	Equation (61) of the Ethereum Yellow Paper states that the expected nonce of the transaction sender is incremented when a transaction is executed \cite{yellowpaper}.
	Thus, the same transaction $ T $ cannot be included more than once in the same block and cannot be included in two different blocks of the same chain.
\end{proof}

\setlemmanumberto{lem:n-1-geq-quorum}
\begin{lemma}
	$n-1 \geq \quorum{n}$ for any $n \geq 4$.
\end{lemma}
\begin{proof}
	The Lemma is proved by the following chain of inequalities and equalities:
	\begin{align}
	\quorum{n} &\equiv 2 \cdot \left \lfloor \frac{n-1}{3}  \right \rfloor  + 1\\
	&\leq 2 \cdot \frac{n-1}{3} + 1 \nonumber \\
	&=\frac{2n+1}{3} \nonumber
	\end{align}
	It is easy to prove that $\frac{2n+1}{3} \leq n-1$ for any $n\geq 4$. 
\end{proof}

\setlemmanumberto{lem:n->2-f(n)}
\begin{lemma}
	$\forall n \geq 1: n > 2\cdot f(n) $
\end{lemma}

\begin{proof}
	\begin{align}
	2 \cdot f(n) &\equiv 2 \left \lfloor \frac{n-1}{3} \right \rfloor \\
	& \leq 2 \frac{n-1}{3} \nonumber \\
	& = \frac{2n}{3} - \frac{2}{3} \nonumber
	\end{align}
	It is easy to see the $ \frac{2n}{3} - \frac{2}{3} < n$ for any $n \geq 1$.
\end{proof}

\section{Proofs for \Cref{sec:safety-analysis-of-the-first-solution}}\label{ap:proofs-for-safety-analysis-of-the-first-solution}
\setlemmanumberto{lem:safety-solution-forward}
\begin{lemma}
	For any $n_h \geq 1$, provided that the number of Byzantine validators is $\leq f(n_h)$, the intersection of any two sets of \quorumoptimal($ n_h $) validators is guaranteed to include an honest validator.
\end{lemma}
\begin{proof}
	Let us assume that $V$ and $W$ are any two sets of validators of size $ \quorumoptimal(n_h) $.
	Let $H$ denote the set of honest validators included in the intersection of $V$ and $W$.
	The minimum size of $H$ is expressed by the following equation:
	\begin{align}
	\min(\| H \|) = \max\big(2 \cdot \quorumoptimal(n_h) -n_h -f(n_h),0\big)
	\end{align}
	
	The following chain of equalities and inequalities proves that ${\min(\| H \|) \geq \frac{1}{3}}$:
	\begin{align}
	\min(\| H \|) &=2 \cdot \quorumoptimal(n_h) -n_h -f(n_h) \\
	&= 2  \left  \lceil \frac{2n_h}{3} \right \rceil - \left \lfloor \frac{n_h-1}{3} \right \rfloor -n \nonumber
	\\
	&\geq 2  \cdot \frac{2n_h}{3}  -  \frac{n_h-1}{3}  -n_h 
	\nonumber \\
	&=\frac{1}{3} \nonumber
	\end{align}
	Since $\min(\| H \|)$ is an integer value, the inequality $\min(\| H \|) \geq \frac{1}{3}$ implies that $\min(\| H \|) \geq 1$.
\end{proof}

\section{Proofs for \Cref{sec:liveness-impact-of-modification-ibftmone}}\label{ap:proof-for-liveness-impact-of-modification-ibftmone}
\setlemmanumberto{lem:qopt+f-leq-n}
\begin{lemma}
	$\forall n > 0: \quorumoptimal(n) + f(n)\leq n $
\end{lemma}

\begin{proof}
	\begin{sloppypar}
		The following chain of equalities and inequalities proves that ${\quorumoptimal(n) + f(n) \leq n +\frac{1}{3}}$:
	\end{sloppypar}
	\begin{align}
	\quorumoptimal(n) + f(n) &=  \left  \lceil \frac{2n}{3} \right \rceil + \left \lfloor \frac{n-1}{3} \right \rfloor \\
	&= \left  \lfloor \frac{2n-1}{3} \right \rfloor + 1 + \left \lfloor \frac{n-1}{3} \right \rfloor \nonumber \\
	&\leq  \frac{2n-1}{3}  + 1 + \frac{n-1}{3} \nonumber \\
	&=n+\frac{1}{3} \nonumber
	\end{align}	
	\begin{sloppypar}
		Since $\quorumoptimal(n) + f(n)$ is an integer value, the inequality ${\quorumoptimal(n) + f(n) \leq n +\frac{1}{3}}$ implies that $\quorumoptimal(n) + f(n)\leq n$.
	\end{sloppypar}
\end{proof}

\section{Proofs for \Cref{sec:analysis-of-the-ibft-and-ibftmone-protocols}}\label{ap:proofs-for-analysis-of-the-ibft-and-ibftmone-protocols}
\setlemmanumberto{lem:n-q<q}
\begin{lemma}
	$ \forall n\geq 4 \land n \neq 6: n - \quorum{n} < \quorum{n}$
\end{lemma}
\begin{proof}
	It is easy to see that the following inequality implies the Lemma:
	\begin{align}
	&2\cdot\quorum{n}>n, \text{ if } n \geq 4 \land n \neq 6 \nonumber
	\end{align}
	We now prove the Lemma by showing that (a) ${2\cdot\quorum{n}\geq 4\cdot \frac{n}{3}-2}$ and (b) $4\cdot \frac{n}{3}-2>n$ when $n \geq 4 \land n \neq 6$.
	\begin{sloppypar}
		The following chain of equalities and inequalities proves (a).
	\end{sloppypar}
	\begin{align}
	2\cdot\quorum{n} &= 4\cdot f(n)+2 \\
	&=4 \left( \left \lceil \frac{n}{3} \right \rceil -1 \right) +2 \nonumber \\
	&=4\left \lceil  \frac{n}{3} \right \rceil -2 \nonumber \\
	&\geq 4\cdot \frac{n}{3}-2 \nonumber
	\end{align}
	It is easy prove that $4\cdot \frac{n}{3}-2>n$ when $n>6$.
	Therefore the Lemma is proved for $n>6$.\\
	For $n=4$ or $n=5$, we have $2\cdot \quorum{n} = 6 > n$.
	This concludes the proof.
\end{proof}

\setlemmanumberto{lem:n-qopt<qopt}
\begin{lemma}
	$\forall n > 0: n - \quorumoptimal(n) < \quorumoptimal(n)$
\end{lemma}
\begin{proof}
	It is easy to see that the following inequality implies the Lemma:
	\begin{align}
	&2\cdot \quorumoptimal(n) > n \nonumber
	\end{align}
	
	The  Lemma is proved by the following sequence of inequalities and equalities:
	\begin{align}
	2\cdot\quorumoptimal(n) &= 2 \cdot \left \lceil \frac{2n}{3} \right \rceil  \\
	&\geq 2\cdot  \frac{2n}{3}  \nonumber \\
	&=\frac{4n}{3} \nonumber \\
	&>n \nonumber
	\end{align}
\end{proof}

\end{document}